\documentclass[a4paper,USenglish]{article}
\usepackage[utf8]{inputenc}
\usepackage{amsmath}
\usepackage{amsfonts}
\usepackage{amsthm}
\usepackage{xcolor,graphicx}
\usepackage{tabularx}
\usepackage{booktabs}
\usepackage{multirow}
\usepackage{pdflscape}
\usepackage{xspace}
\usepackage{a4wide}
\usepackage{dsfont}
\usepackage{nicefrac}

\usepackage{ifthen}

\usepackage{lscape}

\usepackage{rotating}

\usepackage[square,sectionbib,numbers]{natbib}

\usepackage[vlined,linesnumbered,ruled]{algorithm2e}
\SetEndCharOfAlgoLine{}
\SetCommentSty{textrm}
\SetKwProg{Fn}{Function}{}{}

\usepackage[rightcaption]{sidecap}
\sidecaptionvpos{figure}{c}

\usepackage{tikz}
\usetikzlibrary{matrix,arrows,decorations.pathmorphing,decorations.pathreplacing,backgrounds,positioning,fit,shapes,arrows,calc,patterns,snakes}

\tikzstyle{knoten}=[circle,minimum size=4pt,draw,inner sep=1pt,fill=black]
\tikzstyle{alter}=[circle, minimum size=16pt, draw, inner sep=1pt]
\tikzstyle{majarr}=[draw=black]
\tikzstyle{matched}=[ultra thick]

\usepackage{pgfplots}
\usepackage{siunitx}
\usepackage{pgfplotstable}
\usepgfplotslibrary{statistics}

\pgfplotsset{compat=1.15} %

\definecolor{italyGreen}{RGB}{0, 146, 70}
\definecolor{italyRed}{RGB}{206, 43, 55}

\pgfplotsset{
    discard if not/.style 2 args={
        x filter/.append code={
            \edef\tempa{\thisrow{#1}}
            \edef\tempb{#2}
            \ifx\tempa\tempb
            \else
                \def\pgfmathresult{inf}
            \fi
        }
    },
    discard if/.style 2 args={
        x filter/.append code={
            \edef\tempa{\thisrow{#1}}
            \edef\tempb{#2}
            \ifx\tempa\tempb
                \def\pgfmathresult{inf}
            \else
            \fi
        }
    },
    replace if not/.style n args={3}{
        y filter/.append code={
            \edef\tempa{\thisrow{#1}}
            \edef\tempb{#2}
            \ifx\tempa\tempb
            \else
                \def\pgfmathresult{#3}
            \fi
        }
    },
    replace if/.style n args={3}{
        y filter/.append code={
            \edef\tempa{\thisrow{#1}}
            \edef\tempb{#2}
            \ifx\tempa\tempb
                \def\pgfmathresult{#3}
            \else
            \fi
        }
    },
    filter if empty/.style 2 args={
		x filter/.code={
			\pgfplotstablegetelem{\coordindex}{#1}\of{#2}
			\edef\tempb{}
			\ifx\pgfplotsretval\tempb
				\def\pgfmathresult{}
			\else%
			\fi
		},
	}
}

\makeatletter
\pgfplotstableset{
    discard if not/.style 2 args={
        row predicate/.code={
            \def\pgfplotstable@loc@TMPd{\pgfplotstablegetelem{##1}{#1}\of}
            \expandafter\pgfplotstable@loc@TMPd\pgfplotstablename
            \edef\tempa{\pgfplotsretval}
            \edef\tempb{#2}
            \ifx\tempa\tempb
            \else
                \pgfplotstableuserowfalse
            \fi
        }
    },
    discard if/.style 2 args={
        row predicate/.code={
            \def\pgfplotstable@loc@TMPd{\pgfplotstablegetelem{##1}{#1}\of}
            \expandafter\pgfplotstable@loc@TMPd\pgfplotstablename
            \edef\tempa{\pgfplotsretval}
            \edef\tempb{#2}
            \ifx\tempa\tempb
                \pgfplotstableuserowfalse
            \else
            \fi
        }
    }
}
\makeatother

\makeatletter
\pgfplotsset{
    boxplot prepared from table/.code={
        \def\tikz@plot@handler{\pgfplotsplothandlerboxplotprepared}%
        \pgfplotsset{
            /pgfplots/boxplot prepared from table/.cd,
            #1,
        }
    },
    /pgfplots/boxplot prepared from table/.cd,
        table/.code={\pgfplotstablecopy{#1}\to\boxplot@datatable},
        row/.initial=0,
        make style readable from table/.style={
            #1/.code={
                \pgfplotstablegetelem{\pgfkeysvalueof{/pgfplots/boxplot prepared from table/row}}{##1}\of\boxplot@datatable
                \pgfplotsset{boxplot/#1/.expand once={\pgfplotsretval}}
            }
        },
        make style readable from table=lower whisker,
        make style readable from table=upper whisker,
        make style readable from table=lower quartile,
        make style readable from table=upper quartile,
        make style readable from table=median,
        make style readable from table=lower notch,
        make style readable from table=upper notch
}
\makeatother

\usepackage[menucolor=orange!40!black,filecolor=magenta!40!black,urlcolor=blue!40!black,linkcolor=red!40!black,citecolor=green!40!black,colorlinks]{hyperref}

\usepackage[sort&compress,nameinlink,noabbrev,capitalize]{cleveref}

\theoremstyle{definition}
\newtheorem{theorem}{Theorem}[section]
\newtheorem{lemma}[theorem]{Lemma}

\newtheorem{observation}[theorem]{Observation}
\newtheorem{rrule}{Reduction Rule}[section]

\crefname{rrule}{Reduction Rule}{Reduction Rules}
\Crefname{rrule}{RR}{RR}
\Crefname{observation}{Obs}{Obs}
\crefname{observation}{Observation}{Observations}
\crefname{cond}{Condition}{Conditions}
\crefname{step}{Step}{Steps}

\newtheorem{definition}[theorem]{Definition}

\theoremstyle{remark}

\newtheorem{remark}{Remark}
\theoremstyle{plain}

\newcommand{\problemdef}[3]{
	\smallskip
	\begin{center}
		\begin{minipage}{0.95\textwidth}
			\textsc{#1}
			
			\vspace{3pt}
			
			\setlength{\tabcolsep}{3pt}
			\begin{tabularx}{\textwidth}{@{}lX@{}}
				\textbf{Input:} 	& #2 \\
				\textbf{Question:} 	& #3
			\end{tabularx}
		\end{minipage}
	\end{center}
	\smallskip
}

\newcommand{\CMatch}{\textsc{Maximum-Cardinality Matching}\xspace}
\newcommand{\WMatch}{\textsc{Maximum-Weight Matching}\xspace}
\newcommand{\VC}{\textsc{Vertex Cover}\xspace}
\newcommand{\paramEnv}[1]{#1}
\newcommand{\match}{\ensuremath{\text{mm}}}

\newcommand{\solSize}{s}

\newcommand{\maxpath}{maximal path}
\newcommand{\maxpaths}{maximal paths}
\newcommand{\pendcycle}{pending cycle}
\newcommand{\pendcycles}{pending cycles}

\newcommand{\N}{\mathds{N}}

\pgfplotstableread[col sep = comma]{data/old-unweigted-results-changed-permutation.csv}\dataUnweightedPerm

\pgfplotstableread[col sep = comma]{data/results-interesting.csv}\boxPlotsTab
\newcommand{\resultsAllTab}{data/results.csv}
\newcommand{\dataHighWeighted}{data/old-high-weighted-all-results-changed-permutation.csv}
\pgfplotstableread[col sep = comma]{data/results.csv}\namesTable
\pgfplotstablegetrowsof{\namesTable}
\pgfmathtruncatemacro\TotalRowsResultsAll{\pgfplotsretval-1}
\def\zeroOffset{1}

\begin{document}

\title{Data Reduction for Maximum Matching on Real-World Graphs: Theory and Experiments%
\thanks{This work was partially supported by the DFG project FPTinP (NI 369/16).}}
\author{Tomohiro~Koana \and Viatcheslav~Korenwein \and André~Nichterlein \and Rolf~Niedermeier \and Philipp~Zschoche}
\date{Institut f\"ur Softwaretechnik und Theoretische Informatik,  TU~Berlin, Germany,\\
\texttt{\small \{tomohiro.koana,andre.nichterlein,rolf.niedermeier,zschoche\}@tu-berlin.de}}

\maketitle

\begin{abstract}
Finding a maximum-cardinality or maximum-weight matching in (edge-weighted) undirected graphs is among the most prominent problems of algorithmic graph theory. 
For $n$-vertex and~$m$-edge graphs, the best known algorithms run in $\widetilde{O}(m\sqrt{n})$ time. 
We build on recent theoretical work focusing on linear-time data reduction rules for finding maximum-cardinality matchings and complement the theoretical results by presenting and analyzing (thereby employing the kernelization methodology of parameterized complexity analysis) new (near-)linear-time data reduction rules for both the unweighted and the positive-integer-weighted case. 
Moreover, we experimentally demonstrate that these data reduction rules provide significant speedups of the state-of-the art implementations for computing matchings in real-world graphs: the average speedup factor is~4.7 in the unweighted case and 12.72 in the weighted case. 
\end{abstract}

\section{Introduction}
In their book chapter on matching, \citet{KV18} write that ``matching theory is one of the classical and most important topics in combinatorial theory and optimization''.
Correspondingly, the design and analysis of (weighted) matching algorithms plays a pivotal role in algorithm theory as well as in practical computing.
Complementing the rich literature on matching algorithms (see \citet{CDP19} and \citet{DPS18} for recent accounts, the latter also providing a literature overview), in this work we focus on efficient linear-time data reduction rules that may help to speedup superlinear-time matching algorithms.
Notably, while recent breakthrough results on (weighted) matching (including linear-time approximation algorithms~\cite{DP14}) focus on the theory side, we study theory and practice, thereby contributing to both sides.

To achieve our results, we follow and complement recent purely theoretical work~\cite{MNN20} presenting and analyzing linear-time data reductions for the unweighted case. 
More specifically, on the theoretical side we 
provide and analyze further data reduction rules for the unweighted as well as weighted case.
On the practical side, we demonstrate that these data reduction rules may serve to speedup various matching solvers (including state-of-the-art ones) due to \citet{HSt17}, \citet{KP98}, and \citet{Kol09}.

Formally, we study the following two problems; note that we formulate them as decision problems since this better fits with presenting our theoretical part where we prove kernelization results (thereby employing the framework of parameterized complexity analysis).
However, all our data reduction rules are ``parameter-oblivious'' 
and thus also work and are implemented for the optimization versions where the solution size is not known in advance. %
 
\problemdef{\CMatch}
	{An undirected graph~$G=(V,E)$ and $\solSize \in \N$.}
	{Is there a size-$\solSize$ subset~$M \subseteq E$ of nonoverlapping (that is, pairwise vertex-disjoint) edges?}

\vspace{-4mm}

\problemdef{\WMatch}
	{An undirected graph~$G=(V,E)$, non-negative edge weights~$\omega\colon E \rightarrow \N$, and~$\solSize \in \N$.}
	{Is there a subset~$M \subseteq E$ of nonoverlapping edges of weight~$\sum_{e \in M} \omega(e) \ge \solSize$?}

We remark that all our results extend to the case of rational weights; however, natural numbers are easier to cope with. %

\paragraph{Related work.}
\citet{MV80} were the first to announce an~$O(\sqrt{n}m)$-time algorithm for \CMatch on graphs with~$n$ vertices and~$m$ edges; see \citet{Vaz20} for the details of this algorithm.
This time bound was previously achieved only for bipartite graphs~\cite{HK73}.
While the classic matching algorithm of \citet{HK73} is simple, elegant, and also very efficient in practice (in fact we use it in one kernelization algorithm as subroutine), the algorithm of \citet{MV80} is rather complicated and not (yet) competitive in practice.\footnote{The only implementation of the algorithm of \citet{MV80} we are aware of is due to \citet{HSt17}. This solver was the slowest in our experiments.}
In fact, the fastest solver for \CMatch seems to be still the one by \citet{KP98}, with a worst-case running time of~$O(nm \cdot \alpha(n,m))$ ($\alpha$ denotes the inverse of the Ackermann function). 

The (theoretically) fastest algorithm for \WMatch in sparse graphs is by \citet{DPS18} with a running time of~$O(\sqrt{n}m \log (nN))$ (here~$N$ denotes the largest integer weight).
In practice, the fastest solver we found is due to \citet{Kol09}, which is an implementation of Edmonds' algorithm~\cite{Edm65,Edm65-2} for a perfect matching of minimum cost combined with many heuristic speedups.

Providing parameterized algorithms or kernels for \CMatch has recently gained high interest~\cite{CDP19,KN18,HK19,MNN20,GMN17,FLSPW18,IOO17}.
For \WMatch, however, we are only aware of the work by \citet{IOO17} who provided an algorithm with running time~$O(t (m + n \log n))$ (here~$t$ is the tree-depth of the input graph).

In this work, we transfer some data reduction rules for \textsc{Vertex Cover} to \CMatch.
To this end, we use the algorithm by \citet{IOY14} that in~$O(m\sqrt{n})$ time exhaustively applies an LP-based data reduction rule due to~\citet{NT75}.
We refer to \citet{HLSS20} for a brief overview of practically relevant data reduction rules for \textsc{Vertex Cover}.

Very recently, \citet{KLPU20} provided a fine-tuned implementation of degree-based data reduction rules for \CMatch which is on average three times faster than our implementation when considering the same data reduction rules (see \cref{rule:deg-zero-one-vertices,rule:deg-two-vertices} in \Cref{ssec:low-degree-rules}).

\paragraph{Our contributions.}
We extend kernelization results~\cite{MNN20} for \CMatch and lift them to \WMatch. 
Our data reduction rules for \CMatch are well-known (as crown rule~\cite{JCF04} and LP-based rule~\cite{NT75}) for the NP-hard \textsc{Vertex Cover} problem.
Our theoretical contribution here is to show that the crown rule is also correct for \CMatch.
Moreover, we prove that the exhaustive application of the crown rule and exhaustive application of the LP-based rule lead to the very same graph; thus these two known rules can be seen as equivalent.
We provide algorithms to efficiently apply our data reduction rules (for the unweighted and the weighted case). 
Herein, we have a particular eye on exhaustively applying the data reduction rules in (near) linear time, which seems imperative in an effort to practically improve matching algorithms. 
Hence, our main theoretical contribution lies in developing efficient algorithms implementing the data reduction rules, thereby also showing a purely theoretical guarantee on the amount of data reduction that can be achieved in the worst case (this is also known as kernelization in parameterized algorithmics).
We proceed by implementing and testing the data reduction algorithms for \CMatch and \WMatch, thereby demonstrating their practical effectiveness. 
More specifically, combining them in form of preprocessing with various solvers~\cite{Kol09,HSt17,KP98} yields partially huge speedups on sparse real-world graphs (taken from the SNP library~\cite{snap}).
We refer to \cref{tab:results} for an overview over the various solvers (with the core algorithmic approach they implement) and the speedup factors obtained by apply our data reduction rules as a preprocessing.
\begin{table}
		
	\caption{
		Summary of the speedup factors gained by our kernelization on graphs from the SNP library~\cite{snap} with various solvers. 
		We refer to \cref{sec:experiments} for details.
	}
	\centering
	\begin{tabular}{l l @{\hskip 1cm} r r}  \toprule
		\multicolumn{2}{c}{solver} 						& \multicolumn{2}{c}{speedup}\\
		implemented by				& algorithmic approach by& average 	& median \\ \midrule
		\citet{Kol09} (unweighted) 	& \citet{Edm65,Edm65-2}	& $157.30$ 	& $29.27$ \\
		\citet{HSt17} 				& \citet{MV80}			& 608.79 		& 28.87 \\
		\citet{KP98} 				& \citet{Edm65}		& 4.70 & 2.20 \\ 
		\midrule
		\citet{Kol09} (weighted) 	& \citet{Edm65,Edm65-2}	& 12.72 & 1.40 \\
		\bottomrule
	\end{tabular}
	\label{tab:results}
\end{table}

\paragraph{Notation.}
We use standard notation from graph theory. 
All graphs considered in this work are simple and undirected.
For a graph~$G = (V,E)$, we denote with~$E(G) = E$ the edge set.
For a vertex subset~$V' \subseteq V$, we denote with~$G[V']$ the subgraph induced by~$V'$.
We write $uv$ to denote the edge~$\{u,v\}$ and~$G-v$ to denote the graph obtained from~$G$ by removing~$v$ and all its incident edges.
A \emph{feedback edge set} of a graph~$G$ is a set~$X$ of edges such that~$G-X = (V, E \setminus X)$ is a tree or forest.
The \emph{feedback edge number} denotes the size of a minimum feedback edge set.
A \emph{vertex cover} in a graph is a set of vertices that has a nonempty intersection with each edge in the graph.

A \emph{matching} in a graph is a set of pairwise disjoint edges.
Let~$G$ be a graph and let~$M \subseteq E(G)$ be a matching in~$G$.
We denote by~$\match(G)$ a maximum-cardinality matching respectively a maximum-weight matching in~$G$, depending on whether we have edge weights or not.
If there are edge weights~$\omega\colon E \rightarrow \N$, then for a matching~$M$ we denote by~$\omega(M) := \sum_{e\in M}\omega(e)$ the weight of~$M$.
Moreover, we denote with~$\omega(G)$ the weight of a maximum-weight matching~$\match(G)$, i.\,e.~$\omega(G) := \omega( \match(G) )$.
A vertex~$v \in V$ is called \emph{matched} with respect to~$M$ if there is an edge in~$M$ containing~$v$, otherwise~$v$ is called \emph{free} with respect to~$M$.
If the matching~$M$ is clear from the context, then we omit ``with respect to~$M$''.

\paragraph{Kernelization.}
A \emph{parameterized problem} is a set of instances~$(I,k)$ where~$I \in\Sigma^*$ for a finite alphabet $\Sigma$ and~$k\in \mathbb{N}$ is the \emph{parameter}.
We say that two instances~$(I,k)$ and $(I',k')$ of parameterized problems~$P$ and~$P'$ are \emph{equivalent} if~$(I,k)$ is a yes-instance for~$P$ if and only if~$(I',k')$ is a yes-instance for~$P'$. 
A \emph{kernelization} is an algorithm that, given an instance~$(I,k)$ of a parameterized problem~$P$, computes in polynomial time an equivalent instance~$(I',k')$ of~$P$ (the \emph{kernel}) such that $|I'|+k'\leq f(k)$ for some  computable function~$f$. %
We say that~$f$ measures the \emph{size} of the kernel, and if~$f(k)\in k^{O(1)}$, then we say that $P$~admits a polynomial kernel. 
Typically, a kernel is achieved by applying polynomial-time executable data reduction rules.
We call a data reduction rule~$\mathcal{R}$ \emph{correct} if the new instance~$(I',k')$ that results from applying~$\mathcal{R}$ to~$(I,k)$ is equivalent to~$(I,k)$.
An instance is called \emph{reduced} with respect to some data reduction rule if further application of this rule has no effect on the instance.

\paragraph{Structure of this work.} 
In \Cref{sec:unweighted-kernel,sec:weigted-kernel}, we provide the kernelization results for \CMatch and \WMatch which we experimentally evaluate on real-world data sets in \Cref{sec:experiments}.
In \Cref{sec:unweighted-kernel} we discuss the unweighted case by recalling old and presenting new data reduction rules.
In \Cref{sec:weigted-kernel} we show how to extend some of the data reduction rules presented for \CMatch to \WMatch.
In \Cref{sec:experiments}, we describe our experimental results, discuss the effect of our data reduction rules on state-of-the-art solvers, and evaluate the prediction quality of our theoretical kernelization results.
We conclude in \Cref{sec:conclusion} with a glimpse on future research challenges.

\section{Maximum-Cardinality Matching} \label{sec:unweighted-kernel}

For \CMatch we first recall in \Cref{ssec:low-degree-rules} simple data reduction rules for low-degree vertices due to a classic result of \citet{KS81}.
Then we improve the known kernel-size for \CMatch parameterized by the feedback edge number when only these two data reduction rules are exhaustively applied~\cite{MNN20}.

In \Cref{ssec:crowns-LP}, we discuss the crown data reduction rule (designed for \VC \cite{JCF04}) and 
show that it also works for \CMatch.
To this end, we briefly describe a classic LP-based data reduction due to \citet{NT75}.
It was known that this LP-based data reduction also removes all crowns from the input graph \cite{AI16,IOY14}.
We note that this does not immediately imply that we can use the LP-based data reduction in the context of \CMatch
(the correctness is only known for \VC).
We prove that exhaustively applying the crown data reduction rule is equivalent to ``exhaustively'' applying the LP-based data reduction.
This allows us to use the algorithm of \citet{IOY14} that exhaustively applies the LP-based data reduction in order to remove all crowns and nothing else.

Finally, we show in \cref{ssec:relaxed-crowns} a generalization of the crown data reduction rule.
However, we leave it open how to apply this generalized crown data reduction rule efficiently. 
Note that we have implemented all of these data reduction rules (except the generalized crown data reduction rule); see \cref{sec:experiments} for an evaluation.

\subsection{Removing low-degree vertices}\label{ssec:low-degree-rules}
For \CMatch two simple data reduction rules are due to a classic result of \citet{KS81}. 
They deal with vertices of degree at most two.

\begin{rrule}[\cite{KS81}]\label{rule:deg-zero-one-vertices}
	Let~$v \in V$.
	If~$\deg(v) = 0$, then delete~$v$.
	If~$\deg(v) = 1$, then delete~$v$ and its neighbor, and decrease the solution size~$\solSize$ by one.
\end{rrule}

\begin{rrule}[\cite{KS81}]\label{rule:deg-two-vertices}
	Let~$v$ be a vertex of degree two and let~$u,w$ be its neighbors.
	Then remove~$v$, merge~$u$ and~$w$, and decrease the solution size~$\solSize$ by one.
\end{rrule}

If the degree of the considered vertex~$v$ is zero, then the maximum matching size remains unchanged.
Otherwise, we have $|\match(G)| = |\match(G')|+1$, where~$G'$ is the instance resulting from one application of either \cref{rule:deg-zero-one-vertices} or~\ref{rule:deg-two-vertices}:
When applying \cref{rule:deg-zero-one-vertices}, then $v$~is matched to its only neighbor~$u$.
For \cref{rule:deg-two-vertices} the situation is not so clear as~$v$ is matched to~$u$ or to~$w$ depending on how the maximum-cardinality matching in the rest of the graph looks like.
Thus, one can only fix the matching edge with endpoint~$v$ (in the original graph) in a simple postprocessing step.

Each of the above data reduction rules can be exhaustively applied in linear time.
While for \cref{rule:deg-zero-one-vertices} this is easy to see, for \cref{rule:deg-two-vertices} the algorithm needs further ideas~\cite{BK09a}.
However, to exhaustively apply both data reductions rules together only algorithms with superlinear running times are known \cite{BK20}. 

Using the above data reduction rules, one can show a kernel with respect to the parameter \paramEnv{feedback edge number}, that is, the size of a minimum feedback edge set.
We refer to \cref{sec:eval-kernel} for a practical evaluation of a theoretical upper bound on the kernel size.

\begin{theorem}[\cite{MNN20}]\label{thm:fes-lin-kernel}
	\CMatch{} admits a linear-time computable kernel with at most $2k-1$ vertices and at most $3k-2$ edges,
	when parameterized by the \paramEnv{feedback edge number}~$k$. 
\end{theorem}
\noindent
\citet{MNN20} originally proved a kernel with at most $12k$ vertices and $13k$ edges.
However, we can tighten this upper bound in the following way.
\begin{proof}[Proof of \cref{thm:fes-lin-kernel}]
	Our kernelization procedure consists of the following two steps:
	\begin{enumerate}
		\item Apply \cref{rule:deg-zero-one-vertices} exhaustively in linear time. \label{step:deg-one}
		\item Apply \cref{rule:deg-two-vertices} exhaustively in linear time \cite{BK09a}. \label{step:deg-two}
	\end{enumerate}
	Let~$G^{(1)} = (V^{(1)}, E^{(1)})$ and~$G^{(2)} = (V^{(2)}, E^{(2)})$ be the graphs obtained after Steps \eqref{step:deg-one} and \eqref{step:deg-two}, respectively.
	Thus, $G^{(2)}$ is the graph returned by the kernelization algorithm. 
	Note that~$G^{(2)}$ might contain isolated vertices and degree-one vertices.
	
	To evaluate the size of~$G^{(2)}$, we first analyze the structure of~$G^{(1)}$.
	Since~$G^{(1)}$ is an induced subgraph of the input graph~$G$, it follows that~$G^{(1)}$ has a feedback edge set $F^{(1)} \subseteq E^{(1)}$ of size at most $k$.
	We show that~$G^{(1)}$ contains at most~$2k-1$ vertices of degree at least three.
	To this end, consider the graph~$G^{(1)}_F=(V^{(1)}_F,E^{(1)}_F)$ obtained from~$G^{(1)}$ as follows.
	Remove the edges in $F^{(1)}$.
	For each edge $e = uw \in F^{(1)}$, we remove $uw$ and we introduce two degree-one vertices $uv_u^{e}$ and $wv_w^{e}$ adjacent to $u$ and $v$, respectively.
	Formally, we have
	\[V^{(1)}_F := V^{(1)} \cup \left\{ v_u^{e},v_w^{e} \mid e=uw \in F^{(1)} \right\} \;\;\text{and}\;\; E^{(1)}_F := (E^{(1)} \setminus F^{(1)}) \cup \left\{ uv_u^{e}, wv_w^{e} \mid  e=uw \in F^{(1)} \right\}.\]
	Observe that~$G^{(1)}_F$ is a forest where~$V^{(1)}_F \setminus V^{(1)}$ are the leaves and~$V^{(1)} \subseteq V^{(1)}_F$ are the internal vertices.
	Since~$|F^{(1)}| \le k$, we have at most~$2k$ leaves.
	Since~$G^{(1)}_F$ is a forest, it follows that~$G^{(1)}_F$ has at most~$2k-1$ vertices degree of at least three.
	
	We partition the vertex set of~$G^{(1)}$ into $V^{(1)} = V^{(1)}_2 \cup V^{(1)}_{\geq 3}$, where $V^{(1)}_2$ are the vertices of degree two and $V^{(1)}_{\geq 3}$ are the vertices of degree at least three.
	Since \cref{rule:deg-zero-one-vertices} is exhaustively applied, the minimum degree in~$G^{(1)}$ is at least two.
	Moreover, note that by construction of~$G^{(1)}_F$, it follows that the set of vertices with degree at least three is identical in~$G^{(1)}_F$ and~$G^{(1)}$.
	Thus, $|V^{(1)}_{\geq 3}| \le 2k-1$.
	Note that exhaustively applying \cref{rule:deg-two-vertices} in Step~\eqref{step:deg-two} on $G^{(1)}$ will remove all vertices in~$V^{(1)}_2$ (and possibly some of $V^{(1)}_{\geq 3}$).
	Thus, the resulting graph~$G^{(2)}$ of our kernelization procedure has at most $2k-1$ vertices and consequently at most $3k-2$ edges.
\end{proof}
Applying the~$O(m\sqrt{n})$-time algorithm for \CMatch~\cite{MV80} altogether yields an~$O(n + m + k^{1.5})$-time algorithm, where~$k$ is the feedback edge number.

\subsection{Crown and LP-based data reduction}\label{ssec:crowns-LP}
Crowns are a classic data reduction tool for \VC (given an undirected graph find a smallest set of vertices that covers all edges) and can be seen as a generalization of \cref{rule:deg-zero-one-vertices} \cite{JCF04}.
A crown satisfies the following properties (see \cref{fig:crown} for a visualization):

\begin{definition}\label{def:crown}
	A \emph{crown} in a graph~$G = (V,E)$ is a pair $(H,I)$ such that
	\begin{enumerate}
		\item $I \subseteq V$ is an independent set in~$G$ (no two vertices in~$I$ are adjacent in~$G$), \label{crown:IS}
		\item $H = N(I) := \bigcup_{v \in I} N(v)$, and \label{crown:Neighbors}
		\item there is a matching~$M_{H,I}$ between~$H$ and~$I$ that matches all vertices in~$H$. \label{crown:Matching}
	\end{enumerate}
\end{definition}

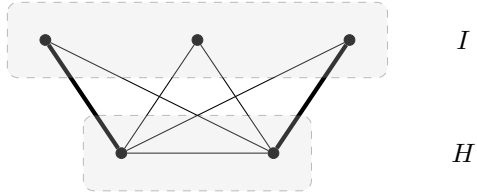
\begin{figure}
	\centering
	\begin{tikzpicture}[auto]

		\node[knoten] at (0,1.5) (v) {};
		\node[knoten] at (1,0) (u) {};
		\node[knoten] at (2,1.5) (a) {};
		\node[knoten] at (4,1.5) (b) {};
		\node[knoten] at (3,0) (c) {};

		\draw[majarr,matched] (v) edge (u);
		\draw[majarr] (v) edge (c);
		\draw[majarr] (u) edge (a);
		\draw[majarr] (c) edge (a);
		\draw[majarr] (u) edge (b);
		\draw[majarr] (u) edge (c);
		\draw[majarr,matched] (b) edge (c);
		
		\node at (5.5,1.5) (v) {$I$};
		\node at (5.5,0) (v) {$H$};
		
		\path[fill=black!15,draw=black,dashed,opacity=.25,rounded corners] (-0.5,2) rectangle (4.5,1);
		\path[fill=black!15,draw=black,dashed,opacity=.25,rounded corners] (0.5,0.5) rectangle (3.5,-0.5);
	\end{tikzpicture}
	\caption{An example of a crown where the vertices in~$H$ but not the ones in~$I$ might have further neighbors in the graph.}
	\label{fig:crown}
\end{figure}

It is not hard to see that, given a crown~$(H,I)$, there is a minimum vertex cover containing all vertices in~$H$: 
The matching~$M_{H,I}$ implies that a minimum size vertex cover in~$G[H \cup I]$ has size~$|M_{H,I}| = |H|$. 
Since~$I$ is an independent set and~$H = N(I)$, taking all vertices in~$H$ into a vertex cover is at least as good as taking some vertices of~$I$.
Thus, we end up with the following data reduction rule.

\begin{rrule}[\cite{JCF04}]\label{rule:crown}
	Let~$(H,I)$ be a crown.
	Then remove all vertices in~$H \cup I$ and decrease the solution size~$s$ by~$|H|$.
\end{rrule}

Luckily, \cref{rule:crown} not only works for \VC but also for \CMatch: 
Simply adding to any maximum-cardinality matching in the reduced graph the matching~$M_{H,I}$ results in a maximum-cardinality matching for~$G$, as proven in the next lemma.

\begin{lemma}
	\cref{rule:crown} is correct, that is, $|\match(G)| = |\match(G')| + |H|$.
\end{lemma}

\begin{proof}
	Let~$(H,I)$ be a crown in the input graph~$G$ and~$G' := G - (H \cup I)$.
	Observe that $\match(G') \cup M_{H,I}$ is a matching of cardinality~$|\match(G')| + |H|$ in~$G$, thus~$|\match(G)| \ge |\match(G')| + |H|$.
	Conversely, observe that~$|\match(G)| \le |\match(G - H)| + |H|$ as each vertex in~$H$ can be matched at most once.
	However, we have~$\match(G - H) = \match(G - (H \cup I)) = \match(G')$ as each vertex in~$I$ has degree zero in~$G-H$ and can thus be removed (see \cref{rule:deg-zero-one-vertices}).
	Thus, $|\match(G)| = |\match(G')| + |H|$.
\end{proof}

Now that we established that \cref{rule:crown} can also be applied for \CMatch, it remains to do so as fast as possible.
However, to find a crown~$(H,I)$ we need to also find a matching~$M_{H,I}$.
Moreover, the size of~$M_{H,I}$ depends on the size of the crown;
a crown can be quite large 
(consider for example a complete bipartite graph~$K_{n,n}$: there is only one crown which is the whole graph).
Thus, applying \cref{rule:crown} even once in linear time seems hard to do. 
Indeed, the best known algorithm to apply \cref{rule:crown} (even just once) runs in~$O(\sqrt{n} m)$ time~\cite{IOY14}.

Since we can compute in~$O(\sqrt{n} m)$ time also a maximum-cardinality matching for the input graph, 
\cref{rule:crown} seems to be not useful (why decrease the size of the input when one can solve the problem in the same time?).
However, the algorithm of \citet{IOY14} to exhaustively apply \cref{rule:crown} has only a single step that requires superlinear time: 
the computation of one maximum-cardinality matching in a \emph{bipartite} graph with~$O(n+m)$ vertices and edges.
To find such a matching, we used a straightforward implementation of the classic algorithm of \citet{HK73}.
It turned out in our experiments that even this implementation is faster than computing a maximum-cardinality matching in non-bipartite graphs with any of the implementations for \CMatch that we tested.

\paragraph{Exhaustively removing crowns.}
Subsequently, we briefly sketch the algorithm of \citet{IOY14} and prove that it exhaustively applies \cref{rule:crown} in an input graph~$G=(V,E)$.
Note that their algorithm efficiently applies a classic linear programming (LP-) based data reduction rule for \VC~\cite{NT75}. 
We first observe that a straight-forward adaptation of this LP-based data reduction rule to \CMatch (working with the LP-relaxation for the \CMatch-ILP) does not seems to work (\cref{fig:lp-matching-counterexample} provides a counterexample).
However, \citet{AI16} already observed that the LP-based data reduction rule of \citet{IOY14} (working with the LP-relaxation for the \VC-ILP) also removes all crowns from the input graph. 
It was hence already clear that the LP-based data reduction rule is at least as powerful as exhaustively applying the crown data reduction rule (in the context of \VC).
Below, we show that exhaustively applying the crown data reduction rule is exactly as powerful as the LP-based data reduction rule.
To be precise, we prove that ``exhaustively'' applying the LP-based data reduction rule (as \citet{IOY14} did) and exhaustively applying the crown data reduction rule (\cref{rule:crown}) results in exactly the same graph.
Thus, we can use the algorithm of \citet{IOY14} working with the LP-relaxation for the \VC-ILP to apply \cref{rule:crown} in the context of \CMatch.

We start with explaining the LP-based kernelization for \VC (refer to \citet[Chapter 2]{CFK+15} for a more detailed description with all proofs).
The standard \emph{integer} linear program (ILP) for \VC is as follows.
\begin{align*}
	\text{Minimize} 	&& \sum_{v \in V} x_v \\
	\text{Subject to} 	&& x_v + x_u & \ge 1 	& \forall uv \in E \\
					&& x_v & \in \{0,1\} 	& \forall v \in V
\end{align*}
As usual, the LP relaxation (subsequently called VC-LP) is obtained by replacing the constraints~$x_v \in \{0,1\}$ by~$0 \le x_v \le 1$.
This LP and its dual LP always have half-integral solutions, that is, there is always an optimal solution such that each variable is assigned a value in~$\{0,\nicefrac{1}{2},1\}$ (this is true for a more general class of LPs called BIP2~\cite{Hoc02}).
Given a half-integral solution for the VC-LP, define the following three sets of vertices corresponding to variables set to~$0$, $\nicefrac{1}{2}$, and~$1$, respectively:
\begin{itemize}
	\item $V_1 := \{v \mid x_v = 1\}$,
	\item $V_{\nicefrac{1}{2}} := \{v \mid x_v = \nicefrac{1}{2}\}$, and
	\item $V_{0} := \{v \mid x_v = 0\}$.
\end{itemize}
The following classic result of \citet{NT75} forms the basis for the LP-based kernelization for \VC.
\begin{theorem}[\cite{NT75}]\label{thm:NT}
	There is a minimum vertex cover~$S$ for~$G$ such that~$V_1 \subseteq S \subseteq V_1 \cup V_{\nicefrac{1}{2}}$.
\end{theorem}

\begin{remark}
The dual of the VC-LP is the following relaxation of the ILP for \CMatch:
\begin{align*}
	\text{Maximize} 	&& \sum_{uv \in E} y_{uv} \\
	\text{Subject to} 	&& \sum_{u \in N(v)} y_{uv} & \le 1 	& \forall v \in V \\
					&& 0 \le y_{uv} & \le 1 				& \forall uv \in E
\end{align*}
Notably, a statement similar to \cref{thm:NT} does not hold for \CMatch, see \cref{fig:lp-matching-counterexample} for two counterexamples.
\newcommand{\drawTikzClique}[4]{%
	\foreach \i in {1,...,#2} {
		\node[knoten] (#1-\i) at ({#3 * cos(#4 + 360 * \i / #2 - 360 / #2))},{#3 * sin(#4 + 360 * \i / #2 - 360 / #2))}) {};
	}
	\foreach \i in {2,...,#2} {
		\pgfmathtruncatemacro{\runs}{\i - 1};
		\foreach \j in {1,...,\runs} {
			\draw[majarr] (#1-\j) edge (#1-\i);
		}
	}
}%
\newcommand{\drawTikzWeightedCycle}[5]{%
	\foreach \i in {1,...,#2} {
		\node[knoten] (#1-\i) at ({#3 * cos(#4 + 360 * \i / #2 - 360 / #2))},{#3 * sin(#4 + 360 * \i / #2 - 360 / #2))}) {};
	}
	\foreach \i in {2,...,#2} {
		\pgfmathtruncatemacro{\ii}{\i - 1};
		\draw[majarr] (#1-\i) edge node{#5} (#1-\ii);
	}
	\draw[majarr] (#1-1) edge node{#5} (#1-#2);
}%
\begin{figure}
	\centering
	\begin{tikzpicture}[auto,scale = 0.75]
		\drawTikzWeightedCycle{L}{3}{1}{180}{0.5}
		\begin{scope}[xshift = 6cm]
			\drawTikzWeightedCycle{R}{3}{1}{0}{0.5}
		\end{scope}
		\node[knoten] (L-0) at (2,0) {};
		\node[knoten] (R-0) at (4,0) {};
		\draw[majarr] (L-0) edge node{1} (R-0);
		\draw[matched] (L-2) edge (L-1);
		\draw[matched] (L-3) edge node{0} (L-0);
		\draw[matched] (R-2) edge (R-1);
		\draw[matched] (R-3) edge node{0} (R-0);
		\draw[majarr] (L-0) edge node{0} (L-2);
		\draw[majarr] (R-0) edge node{0} (R-2);
		\begin{scope}[xshift=10cm]
			\drawTikzWeightedCycle{L}{5}{1}{0}{.5}
			\begin{scope}[xshift = 4cm]
				\drawTikzWeightedCycle{R}{3}{1}{180}{.5}
			\end{scope}
			\draw[matched] (L-1) edge node{0} (R-1);
			\draw[matched] (L-2) edge (L-3);
			\draw[matched] (L-4) edge (L-5);
			\draw[matched] (R-2) edge (R-3);
		\end{scope}
	\end{tikzpicture}
	\caption{
		Examples in which a solution for the matching-LP does not help in finding a maximum-cardinality matching. 
		In both graphs, bold edges indicate a perfect matching and the numbers next to the edges give a valid optimal solution to the LP for \CMatch.
		In the left graph, the single edge (in the middle), whose variable is set to one, is not contained in any perfect matching.
		In the right graph, there is only one perfect matching that contains one edge (in the middle), whose variable is set to zero.
	}
	\label{fig:lp-matching-counterexample}
\end{figure}
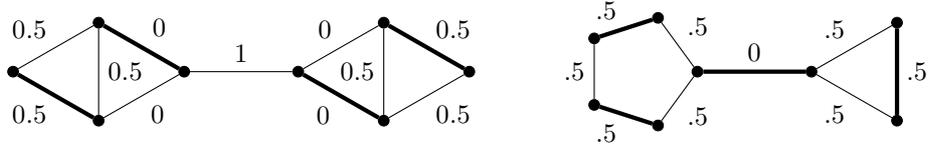%
\end{remark}

The following data reduction rule is an immediate consequence of \cref{thm:NT}:

\begin{rrule}\label{rule:LP}
	Compute a solution for the VC-LP, remove all vertices in~$V_0 \cup V_1$, and decrease the solution size by~$|V_1|$.
\end{rrule}

\citet{AFLS07,CC08} independently showed that \Cref{rule:LP} removes a crown (see \cref{def:crown} for its definition).
We include the proof for the sake of completeness.

\begin{lemma}[\cite{AFLS07,CC08}] \label{lem:lpcrown}
	For any solution $X$ for the VC-LP, $(V_1, V_0)$ is a crown.
\end{lemma}
\begin{proof}
	It is easy to see that~$V_0$ is an independent set since the constraint~$x_v + x_u \ge 1$ for all~$uv \in E$ forbids to set two variables of adjacent vertices to zero.
	Thus, we have that~$N(V_0) \subseteq V_1$.
	It follows from the optimality of~$X$ that~$N(V_0) = V_1$ (any variable of a vertex~$v \in V_1 \setminus N(V_0)$ could be set to~$\nicefrac{1}{2}$ instead).
	To show that there is a matching between~$V_0$ and~$V_1$ that matches all vertices in~$V_1$, 
	we use the optimality of~$X$ together with Hall's marriage 
	theorem\footnote{
		Hall's marriage theorem states that, in any bipartite graph $G=(X \uplus Y,E)$, 
		there is an $X$-saturating matching 
		if and only if 
		for every subset $W \subseteq X$ we have $|W|\leq |N(W)|$.
		In other words: every subset $W \subseteq X$ has sufficiently many adjacent vertices in $Y$.
	}:
	Assuming that there is no such matching, it follows from Hall's marriage theorem that there is a subset~$W \subseteq V_1$ such that~$|N(W) \cap V_0| < |W|$.
	However, this implies that setting all variables corresponding to vertices in~$W \cup (N(W) \cap V_0)$ to $\nicefrac{1}{2}$ results in a better solution to the VC-LP, a contradiction to the optimality of~$X$.
\end{proof}

Subsequently, we first show how to efficiently compute a solution for the VC-LP.
Then, we show that ``exhaustively'' applying \cref{rule:LP} 
and ``exhaustively'' applying \cref{rule:crown} leads to exactly the same kernel.

A half-integral solution for the VC-LP can be obtained by computing a minimum vertex cover in a bipartite graph~$\overline{G}$:
The vertices~$\overline{V}$ of~$\overline{G}$ consist of two copies of~$V$, called~$V_L$ and~$V_R$; for~$i \in \{L,R\}$ we set $V_i := \{v_i \mid v \in V\}$ and~$\overline{V} := V_L \cup V_R$.
The edges~$\overline{E}$ of~$\overline{G}$ are as follows: $\overline{E} := \{v_Lu_R,v_Ru_L \mid uv \in E\}$.
Thus, $\overline{G}$ has~$2n$ vertices and~$2m$ edges.
Then, a solution for the VC-LP can be constructed from a vertex cover~$\overline{S}$ for~$\overline{G}$ as follows~\cite{NT75}:
$$ x_v = \begin{cases}
			1, & \text{if } v_L \in \overline{S} \wedge v_r \in \overline{S}, \\
			0, & \text{if } v_L \notin \overline{S} \wedge v_r \notin \overline{S}, \\
			\nicefrac{1}{2}, & \text{else.}
         \end{cases}$$ 
	 Hence, using Kőnig's theorem\footnote{Kőnig's theorem states that, in any bipartite graph, the number of edges in a maximum matching is equal to the number of vertices in a minimum vertex cover.
	 Moreover, such a minimum vertex cover can be constructed in linear time given a maximum matching.},
	 we can compute a solution for the VC-LP in~$O(\sqrt{n}m)$ time with the algorithm of \citet{HK73}.

Of course, to have \cref{rule:LP} as strong as possible, we would like to have a solution for the VC-LP with the maximum number of variables set to~$0$ and~$1$.
Note that the above solution for the VC-LP does not necessarily fulfill this condition.
However, \citet{IOY14} provided an algorithm that, given any solution for the VC-LP, computes an optimal solution for the VC-LP that minimizes, over all half-integral optimal solutions, the number of variables set to~$\nicefrac{1}{2}$ in \emph{linear time}.
Hence, using the algorithm of \citet{IOY14}, we can exhaustively apply \cref{rule:LP} in~$O(\sqrt{n} m)$ time.
By \emph{exhaustively applying \cref{rule:LP}}, we mean that we apply \cref{rule:LP} until the solution setting all variables to~$\nicefrac{1}{2}$ is the unique optimal solution for the VC-LP.

It remains to show that exhaustively applying \cref{rule:LP} results in the same instance as exhaustively applying \cref{rule:crown}.
This extends the work of \citet{AFLS07} who showed that applying \cref{rule:LP} exhaustively removes all crowns.

\begin{lemma}
	\label{lem:LP-equivalent-to-crown}
	Let~$G$ be a graph, let~$G_{\text{LP}}$ be the graph obtained from exhaustively applying \cref{rule:LP} on~$G$, and let~$G_{\text{crown}}$ be the graph obtained from exhaustively applying \cref{rule:crown}.
	Then~$G_{\text{LP}} = G_{\text{crown}}$.
\end{lemma}
\begin{proof}
	First we show that there exists a set $V_{\nicefrac{1}{2}}^{\star}$ of vertices such that for every half-integral optimal solution for the VC-LP for~$G$ that minimizes the number of variables set to~$\nicefrac{1}{2}$, the set of vertices whose variables are set to $\nicefrac{1}{2}$ is $V_{\nicefrac{1}{2}}^{\star}$.
	Let $X'$ and $X''$ be two half-integral solutions for the VC-LP with the minimum number of variables set to $\nicefrac{1}{2}$.
	For $i \in \{ 1, 0, \nicefrac{1}{2} \}$, let~$V_i'$ and $V_i''$ be, as defined above, the set of vertices whose variables are set to~$i$ in $X'$ and $X''$, respectively.
	Assume for the sake of contradiction that $V_{\nicefrac{1}{2}}' \ne V_{\nicefrac{1}{2}}''$.
	We claim that the following (call it $X$) is an optimal solution for the VC-LP:
	\begin{align*}
		x_v = \begin{cases}
			1 & \text{if } v \in V_1' \cup (V_{\nicefrac{1}{2}}' \cap V_1''), \\
			0 & \text{if } v \in V_0' \cup (V_{\nicefrac{1}{2}}' \cap V_0''), \\
			\nicefrac{1}{2} & \text{else } (\text{that is}, v \in V_{\nicefrac{1}{2}}' \setminus (V_1'' \cup V_0'')).
		\end{cases}
	\end{align*}
	We define $V_1$, $V_0$, and $V_{\nicefrac{1}{2}}$ analogously for~$X$.

	To prove the claim, we first verify that $X$ is a solution for the VC-LP.
	It suffices to show that $N(v) \subseteq V_1$ for each vertex $v \in V_0$.
	If $v \in V_0'$, then we have $N(v) \subseteq V_1' \subseteq V_1$.
	Otherwise (that is, $v \in V_{\nicefrac{1}{2}}' \cap V_0''$), we have
	\(
		N(v) \subseteq N(V_{\nicefrac{1}{2}}' \cap V_0'')
		\subseteq N(V_{\nicefrac{1}{2}}') \cap N(V_0'')
		\subseteq (V_1' \cup V_{\nicefrac{1}{2}}') \cap V_1''
		\subseteq V_1
	\),
	because $N(V_{\nicefrac{1}{2}}') \subseteq V_1' \cup V_{\nicefrac{1}{2}}'$ and $N(V_0'') \subseteq V_1''$.
	Thus, we see that $X$ is a solution for the VC-LP.

	We then show that $X$ is optimal for the VC-LP.
	By \Cref{lem:lpcrown}, there is a matching $M$ in $G[V_1' \cup V_0']$ that matches all vertices of $V_1'$, and thereby, we have $\sum_{v \in V_1' \cup V_0'} x_v'' \ge \sum_{uv \in M} x_u'' + x_v'' \ge |M| \ge |V_1'|$.
	It follows that
	\begin{align*}
		\sum_{v \in V} x_v''
		= \sum_{v \in V_1' \cup V_0'} x_v'' + \sum_{v \in V_{\nicefrac{1}{2}}'} x_v''
		&\ge |V_1'| + |V_{\nicefrac{1}{2}}' \cap V_1''| + \frac{1}{2} |V_{\nicefrac{1}{2}}' \setminus (V_0'' \cup V_1'')| \\
		&= |V_1'| + \frac{1}{2} |V_{\nicefrac{1}{2}}'| + \frac{1}{2} (|V_{\nicefrac{1}{2}}' \cap V_1''| - |V_{\nicefrac{1}{2}}' \cap V_0''|).
	\end{align*}
	Note that $\sum_{v \in V} x_v'' = \sum_{v \in V} x_v' = |V_1'| + \frac{1}{2} |V_{\nicefrac{1}{2}}'|$ by the optimality of $X'$ and $X''$.
	Thus, we obtain $|V_{\nicefrac{1}{2}}' \cap V_1''| - |V_{\nicefrac{1}{2}}' \cap V_0''| \le 0$.
	Then the cost of $X$ is
	\begin{align*}
		\sum_{v \in V} x_v
		&= (|V_1'| + |V_{\nicefrac{1}{2}}' \cap V_1''|) + \frac{1}{2} |V_{\nicefrac{1}{2}}' \setminus (V_1'' \cup V_0'')| \\
		&= |V_1'| + \frac{1}{2} |V_{\nicefrac{1}{2}}'| + \frac{1}{2} (|V_{\nicefrac{1}{2}}' \cap V_1''| - |V_{\nicefrac{1}{2}}' \cap V_0''|)
		\le |V_1'| + \frac{1}{2} |V_{\nicefrac{1}{2}}'|,
	\end{align*}
	and hence $X$ is an optimal solution for the VC-LP for $G$.
	Since the number of variables set to $\nicefrac{1}{2}$ in~$X$ is smaller than that of $X'$ and $X''$, this contradicts our assumption on $X'$ and $X''$.
	Consequently,~$G_{\text{LP}}$ is well-defined---$G_{\text{LP}} = G[V_{\nicefrac{1}{2}}^\star]$.

	It remains to show that if \cref{rule:crown} is applied exhaustively, 
	then the VC-LP for the resulting graph~$G_{\text{crown}}$ has a unique optimal solution, in which all variables are set to~$\nicefrac{1}{2}$.
	Assume towards a contradiction that the VC-LP for~$G_{\text{crown}}$ has a solution~$X$ with some variables not set to~$\nicefrac{1}{2}$.
	Then by \Cref{lem:lpcrown}, $G_{\text{crown}}$ has a crown, which is a contradiction.
	Hence, in~$G_{\text{crown}}$ the VC-LP has a unique solution: setting all variables to~$\nicefrac{1}{2}$.
\end{proof}

It is known that applying \cref{rule:LP} results in an instance with at most~$2 \tau$ vertices, where~$\tau$ is the vertex cover number (the size of a minimum vertex cover)~\cite{CFK+15}.
Combining this, the algorithm of \citet{IOY14}, and \cref{lem:LP-equivalent-to-crown} gives the following theorem.

\begin{theorem}\label{thm:crown-exhaustive-application}
	\cref{rule:crown} can be exhaustively applied in~$O(\sqrt{n} m)$ time and the resulting instance has at most~$2\tau$ vertices.
\end{theorem}

We refer to \cref{sec:eval-kernel} for a practical evaluation of the theoretical upper bound of \Cref{thm:crown-exhaustive-application} on the number of vertices in the reduced instance.
However, note that from a theoretical point of view, the kernel given in \cref{thm:fes-lin-kernel} is incomparable to the upper bound of \Cref{thm:crown-exhaustive-application}:
In a large odd cycle~$C_{2n+1}$, applying the kernel behind \cref{thm:fes-lin-kernel} yields a constant-size kernel (as the feedback edge number is one), that is, \cref{thm:fes-lin-kernel} essentially solves this instance.
\cref{rule:crown}, however, does not remove a single vertex in this case.
Conversely, applying \cref{rule:crown} on a complete bipartite graph~$K_{3,n}$ (here~$\tau = 3$) solves the instance again, but the kernel behind \cref{thm:fes-lin-kernel} does not remove a single vertex.

Note that the algorithm behind \cref{thm:crown-exhaustive-application} contains only one step requiring super-linear running time: 
The computation of a bipartite matching to compute an initial solution for the VC-LP.
Since computing matchings in practice is much easier for bipartite than for general graphs 
(even though the known theoretical worst-case running times are the same~\cite{Vaz20}, 
it is not surprising that our implementation that exhaustively applies \cref{rule:crown} is significantly faster than the state-of-the-art matching implementations on the respective input graph (see~\cref{sec:experiments}). 

\subsection{Relaxed Crowns}\label{ssec:relaxed-crowns}
In this subsection, we provide a data reduction rule that generalizes \cref{rule:deg-two-vertices} (for degree-two vertices) in the same sense as \cref{rule:crown} (crown rule) generalizes \cref{rule:deg-zero-one-vertices} (for degree-one vertices).
However, in contrast to \cref{rule:crown} we decided not to implement this rule as we do not have a sufficiently fast algorithm to apply the rule.
Our new rule uses a relaxed crown concept, defined as follows:
\begin{definition}\label{def:generalized-crown}
	A \emph{relaxed crown} in a graph~$G = (V,E)$ is a pair $(H,I)$ such that
	\begin{enumerate}
		\item $I \subseteq V$ is an independent set in~$G$ (no two vertices in~$I$ are adjacent in~$G$), \label{r-crown:IS}
		\item $H = N(I) := \bigcup_{v \in I} N(v)$, and \label{r-crown:Neighbors}
		\item for every~$v \in H$ there is a matching~$M_{H,I,v}$ between~$H\setminus\{v\}$ and~$I$ that matches all vertices in~$H\setminus\{v\}$. \label{r-crown:Matching}
	\end{enumerate}
\end{definition}
Compared to the crown concept we relaxed Condition~\eqref{r-crown:Matching}. 
In particular, note that in the relaxed crown, the set~$H$ can be larger (by one vertex) than~$I$!
Our new data reduction rule is as follows (see \cref{fig:relaxed-crown} for an illustration).
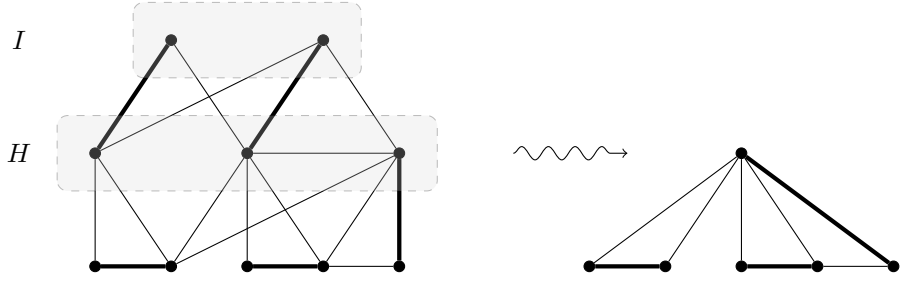
\begin{figure}
	\centering
	\begin{tikzpicture}[auto]
		\foreach \i in {1,2} {
			\node[knoten] at (2*\i - 1,0) (i\i) {};
		}
		\foreach \i in {1,2,3} {
			\node[knoten] at (2*\i - 2,-1.5) (h\i) {};
		}
		\foreach \i in {1,...,5} {
			\node[knoten] at (\i-1,-3) (u\i) {};
		}
	
		\foreach \i in {1,2} {
			\foreach \j in {1,2} {
				\draw[majarr] (i\i) edge (h\j);
			}
		}
		\draw[majarr] (i2) edge (h3);
		\draw[majarr] (h2) edge (h3);
		\draw[majarr,matched] (i1) edge (h1);
		\draw[majarr,matched] (i2) edge (h2);

		\foreach \i / \j in {1/1, 1/2, 2/2, 2/3, 2/4, 3/2, 3/4} {
			\draw[majarr] (h\i) edge (u\j);
		}
		\draw[majarr,matched] (h3) edge (u5);
		\draw[majarr,matched] (u1) edge (u2);
		\draw[majarr,matched] (u4) edge (u3);
		\draw[majarr] (u4) edge (u5);

		\node at (-1,-1.5) (v) {$H$};
		\node at (-1,0) (v) {$I$};
		
		\path[fill=black!15,draw=black,dashed,opacity=.25,rounded corners] (-0.5,-2) rectangle (4.5,-1);
		\path[fill=black!15,draw=black,dashed,opacity=.25,rounded corners] (0.5,0.5) rectangle (3.5,-0.5);
		
		\draw [->,snake=snake,line after snake=1mm] (5.5,-1.5) -- (7,-1.5);

		\begin{scope}[xshift=6.5cm]
			\node[knoten] at (2,-1.5) (w) {};
			\foreach \i in {1,2,3,4,5} {
				\node[knoten] at (\i-1,-3) (u\i) {};
				\draw[majarr] (u\i) edge (w);
			}		
			\draw[majarr,matched] (u5) edge (w);
			\draw[majarr,matched] (u1) edge (u2);
			\draw[majarr,matched] (u4) edge (u3);
			\draw[majarr] (u4) edge (u5);
		\end{scope}
	\end{tikzpicture}
	\caption{Left-hand side: A graph with a relaxed crown~$(H,I)$ and a maximum matching (thick edges) of cardinality five. Right-hand side: The graph obtained by applying \cref{rule:relaxed-crown}. The maximum-cardinality matching highlighted on the right-hand side is by~$|H|-1 = 2$ smaller than the maximum-cardinality matching highlighted on the left-hand side.}
	\label{fig:relaxed-crown}
\end{figure}
\begin{rrule}\label{rule:relaxed-crown}
	Let~$G$ be a crown-free graph and let~$(H,I)$ be a relaxed crown in~$G$.
	Then remove all vertices in~$H \cup I$, add a new vertex~$w$ with~$N(w) = \bigcup_{u \in H} N(u) \setminus (H \cup I)$ and decrease the solution size~$s$ by~$|H|-1$.
\end{rrule}
\begin{lemma}
	\cref{rule:relaxed-crown} is correct.
\end{lemma}
\begin{proof}
	($\Rightarrow$): Let~$M \subseteq E$ be a maximum-cardinality matching in the input graph~$G$ and let~$G'$ be the reduced graph.
	We split~$M$ into three parts~$M = M_0 \cup M_1 \cup M_2$ as follows (edges with zero, one, or two endpoints in the relaxed crown): 
	$M_i := \{uv \in M \mid i = \mid uv \cap (H \cup I)|\}$.
	Since~$I$ is an independent set, it follows that~$|M_1|+|M_2| \le |H|$.
	We now make a case distinction whether or not~$M_1$ is the empty set.
	
	\emph{Case 1.} $M_1 = \emptyset$:
	Note that~$|M_2| \le |H|.$
	By assumption, $G$ does not contain a crown and, thus, we can strengthen this to $|M_2| \le |H|-1$.
	Hence, $M_0$ is a matching for~$G'$ with~$|M_0| = |M| - |M_2| \ge |M| - (|H|-1)$.
	
	\emph{Case 2.} $M_1 \neq \emptyset$:
	Let~$uv \in M_1$ be an arbitrary edge with~$u \in H$.
	Then, $wv$~is in~$G'$ and~$M_0 \cup \{wv\}$ is a matching for~$G'$ of size at least~$|M_0| + 1 = |M| - (|M_1| + |M_2|) + 1  \ge |M| - (|H|-1)$.
	
	($\Leftarrow$): Let~$M'$ be a maximum-cardinality matching in~$G'$.
	If $w$ is not matched in $M'$, then $M \cup M_{H, I, v}$ is a matching of size $|M'| + |H| - 1$ for an arbitrary vertex $v \in H$.
	So we may assume that $wu \in M'$.
	Let $v \in H$ be an arbitrary vertex adjacent to $u$.
	Then, $(M \setminus \{ uw \}) \cup M_{H, I, v}$ is a matching of size~$|M'| + |H| - 1$ for~$G$.
	Hence, $|\match(G)| = |\match(G')| + |H| - 1$.
\end{proof}

\section{Maximum-Weight Matching} \label{sec:weigted-kernel}
In this section, we show how to lift \cref{thm:fes-lin-kernel} (dealing with \CMatch{}) to the weighted case. %
\cref{rule:deg-zero-one-vertices,rule:deg-two-vertices} are based on the simple observation that for every vertex~$v \in V$ of degree at least one, there exists a maximum-cardinality matching containing~$v$:
If~$v$ is not matched, then take an arbitrary neighbor~$u$ of~$v$, remove the edge containing~$u$ from a maximum-cardinality matching, and add the edge~$uv$.
This observation does not hold in the weighted case---see, e.\,g., \Cref{fig:deg-one-rule-weighted} (left-hand side) where the only maximum-weight matching $\{au,bc\}$ leaves~$v$ free. 
Thus, we need new ideas to obtain efficient data reduction rules for the weighted case.

\paragraph{Vertices of degree at most one.}
We start with the simple case of dealing with vertices of degree at most one. 
Here, the following data reduction rule is obvious.

\begin{rrule}
\label{rule:deg-zero-weighted}
	If $\deg(v)=0$ for a vertex $v \in V$, then delete $v$.
	If $\omega(e) = 0$ for an edge $e \in E$, then delete $e$.
\end{rrule}

Next, we show how to deal with degree-one vertices, see \Cref{fig:deg-one-rule-weighted} for a visualization. 

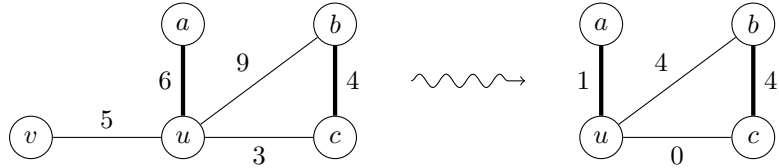
\begin{figure}[t]
	\centering
	\begin{tikzpicture}[auto]
		\node[alter] at (0,0) (v) {$v$};
		\node[alter] at (2,0) (u) {$u$};
		\node[alter] at (2,1.5) (a) {$a$};
		\node[alter] at (4,1.5) (b) {$b$};
		\node[alter] at (4,0) (c) {$c$};

		\draw[majarr] (v) edge node{5} (u);
		\draw[majarr,matched] (u) edge node{6} (a);
		\draw[majarr] (u) edge node{9} (b);
		\draw[majarr] (u) edge node[swap]{3} (c);
		\draw[majarr,matched] (b) edge node{4} (c);

		\draw [->,snake=snake,line after snake=1mm] (5,0.75) -- (6.5,0.75);
		
		\begin{scope}[xshift=5.5cm]
			\node[alter] at (2,0) (u) {$u$};
			\node[alter] at (2,1.5) (a) {$a$};
			\node[alter] at (4,1.5) (b) {$b$};
			\node[alter] at (4,0) (c) {$c$};

			\draw[majarr,matched] (u) edge node{1} (a);
			\draw[majarr] (u) edge node{4} (b);
			\draw[majarr,matched] (b) edge node{4} (c);
			\draw[majarr] (u) edge node[swap]{0} (c);
		\end{scope}
	\end{tikzpicture}
	\caption{Left: Input graph. Right: The graph after applying \cref{rule:deg-one-weighted} to vertex $v$. Bold edges indicate the unique maximum-weight matching in each graph.}
	\label{fig:deg-one-rule-weighted}
\end{figure}

\begin{rrule}
	\label{rule:deg-one-weighted}
	Let $G = (V,E)$ be a graph with non-negative edge weights~$\omega\colon E \rightarrow \N$.
	Let~$v$ be a degree-one vertex and let $u$ be its neighbor. 
	Then delete $v$, set the weight of every edge~$e$ incident with~$u$ to~$\max\{0,\omega(e)-\omega(uv)\}$, and decrease the solution value~$s$ by~$\omega(uv)$.
\end{rrule}

While proving the correctness of this rule (see next lemma) is relatively straightforward, the naive algorithm to exhaustively apply \cref{rule:deg-one-weighted} is too slow for our purpose:
If the edge weights are adjusted immediately after deleting~$v$, then exhaustively applying the rule to a star requires~$\Theta(n^2)$ time.
However, as we subsequently show, \cref{rule:deg-one-weighted} can be exhaustively applied in linear time.

\begin{lemma}
	\label[lemma]{lem:deg-one-weighted-correct}
	\cref{rule:deg-one-weighted} is correct.
\end{lemma}
\begin{proof}
	Let $v$ be a vertex of degree one and let $u$ be its neighbor. 
	Let $M$ be a matching of weight at least $s$ for~$G$. 
	We assume without loss of generality that $M$ is of maximum weight and, hence, $u$ is matched.
	If $uv \in M$, then the deletion of $v$ decreases the weight of the matching by~$\omega(uv)$.
	Hence, the resulting graph $G'$ (with adjusted weights) has a matching of weight at least $s-\omega(uv)$.
	If~$uv \not\in M$, then $M$~is also contained in the resulting graph $G'$.
	As~$v$ is not matched, $M$~contains exactly one edge~$e$ with~$u \in e$. 
	Thus, $e$ has in~$G'$ weight~$\max \{ 0, w(e)-w(uv) \}$ and~$M$ has in~$G'$ weight at least~$s-\omega(uv)$.
	
	Conversely, let~$M'$ be a matching in the reduced graph~$G'$ with weight at least~$s-\omega(uv)$.
	We construct a matching~$M''$ for~$G$ as follows.
	First, consider the case that~$u$ is matched by an edge~$e$.
	If~$e$ has in~$G'$ weight more than zero, then set~$M'' := M'$. %
	If~$e$ has in~$G'$ weight zero, then set~$M'':= (M' \setminus \{e\}) \cup \{uv\}$.
	Second, if~$u$ is free, then set~$M'':= M' \cup \{uv\}$. 
	In all three cases~$M''$ is a matching in $G$ with weight at least~$s$.
\end{proof}

\begin{lemma}
\label[lemma]{lem:deg-one-weighted-time}
	\cref{rule:deg-one-weighted} can be exhaustively applied in $O(n+m)$ time.
\end{lemma}

\begin{proof}
	The basic idea of the algorithm exhaustively applying \cref{rule:deg-one-weighted} in linear time is as follows: 
	We store in each vertex a number indicating the weight of the heaviest incident edge removed due to \cref{rule:deg-one-weighted}.
	Then, whenever we want to access the ``current'' weight of an edge~$e$, then we subtract from~$\omega(e)$ the two numbers stored in the two incident vertices.
	Once \cref{rule:deg-one-weighted} is no more applicable, then we update the edge weights to get rid of the numbers in the vertices in order to create a \WMatch instance.
	
	The details of the algorithm are as follows.
	First, in~$O(n+m)$ time we collect all degree-one vertices in a list~$L$ and initialize for each vertex~$v$ a counter~$c(v) := 0$.
	Then, we process~$L$ one by one.
	For a degree-one vertex~$v \in L$, let~$u$ be its neighbor. 
	We decrease~$s$ by~$\max\{0,w(uv) - c(u) - c(v)\}$, then set~$c(u) := c(u) + \max\{0, w(uv) - c(u) - c(v)\}$, and then delete~$v$.
	If after the deletion of~$v$ its neighbor~$u$ has degree one, then~$u$ is added to~$L$.
	Thus, after at most~$n$ steps, each one doable in constant time, we processed~$L$.
	When~$L$ is empty, then in~$O(m)$ time we update for each edge~$uv$ its weight~$w(uv) := \max\{0,w(uv) - c(u) - c(v)\}$.
	This finishes the description of the  algorithm.
	
	Observe that we have the following invariant when processing the list~$L$: the weight of an edge~$uv$ is~$\max\{0,w(uv) - c(u) - c(v)\}$.
	With this invariant, it is easy to see that the algorithm indeed applies \cref{rule:deg-one-weighted} exhaustively.
\end{proof}

Note that after applying \cref{rule:deg-one-weighted} we can have weight-zero edges and thus \cref{rule:deg-zero-weighted} might become applicable.
We do not know whether \cref{rule:deg-zero-weighted,rule:deg-one-weighted} \emph{together} can be applied exhaustively in linear time. 
However, for the kernel we present at the end of this section it is sufficient to apply \cref{rule:deg-one-weighted} exhaustively.

\paragraph{Vertices of degree two.}
Lifting \cref{rule:deg-two-vertices} to the weighted case is more delicate than lifting \cref{rule:deg-zero-one-vertices} to \cref{rule:deg-zero-weighted,rule:deg-one-weighted}.
The reason is that the two incident edges might have different weights.
As a consequence, we cannot decide locally what to do with a degree-two vertex.
Instead, we process multiple degree-two vertices at once.
To this end, we use the following notation.

\begin{definition}
	\label[definition]{def:maxpath}
	Let~$G$ be a graph. 
	A path~$P = v_0v_1\ldots v_\ell$ is a \emph{\maxpath{}} in~$G$ if $\ell \ge 3$ and the inner vertices~$v_1, v_2, \ldots, v_{\ell-1}$ all have degree two in~$G$, but the endpoints~$v_0$ and~$v_\ell$ do not, that is, $\deg_G(v_1) = \ldots = \deg_G(v_{\ell-1}) = 2$, $\deg_G(v_0) \neq 2$, and $\deg_G(v_\ell) \neq 2$.
\end{definition}

\begin{definition}
	\label[definition]{def:pendcycle}
	Let~$G$ be a graph. 
	A cycle~$C = v_0v_1\ldots v_\ell v_0$ is a \emph{\pendcycle{}} in~$G$ if at most one vertex in~$C$ does \emph{not} have degree two in~$G$.
\end{definition}

The reason to study \maxpaths{} and \pendcycles{} is that 
we can compute a maximum-weight matching on path and cycle graphs in linear time, as stated next.
This allows us to preprocess all vertices in a \maxpath{} or a \pendcycle{} at once.

\begin{observation}
	\label[observation]{obs:path-cycle-lin-time}
	\WMatch can be solved in~$O(n)$ time on paths and cycles.
\end{observation}

\begin{proof}
	If the input graph~$G$ is a path, then by exhaustively applying \cref{rule:deg-zero-weighted,rule:deg-one-weighted}, we can compute a maximum-weight matching.
	Otherwise, if~$G$ is a cycle, then we take an arbitrary edge~$e$ and distinguish two cases.
	First, we take~$e$ into a matching and remove both endpoints from the graph.
	In the resulting path, we compute in linear time a maximum-weight matching~$M$.
	Second, we delete~$e$ and obtain a path for which we compute in linear time a maximum-weight matching~$M'$.
	We then simply choose between~$M \cup\{e\}$ and~$M'$ the heavier matching as the result.
\end{proof}

Now, using \cref{obs:path-cycle-lin-time}, we introduce data reduction rules for \maxpaths{} and \pendcycles{}.
Both rules are based on a similar idea which is easier to explain for a \pendcycle.
Let~$C$ be a \pendcycle{} and~$u\in C$ be the degree-at-least-three vertex in~$C$. 
Then there are two cases: $u$ is matched with a vertex not in~$C$ or it is not.
Let~$M$ be a maximum-weight matching for~$G$, and let~$M'$ be a maximum-weight matching with the constraint that~$u$ is matched to a vertex outside~$C$.
Clearly, $M \cap E(C)$ is at least as large as~$M' \cap E(C)$.
Looking only at~$C$, all that we need to know is the difference of the weights of these two matchings.
This can be encoded with one vertex~$z$ which replaces the whole cycle~$C$ (see \Cref{fig:rrule cycle3} for an illustration).
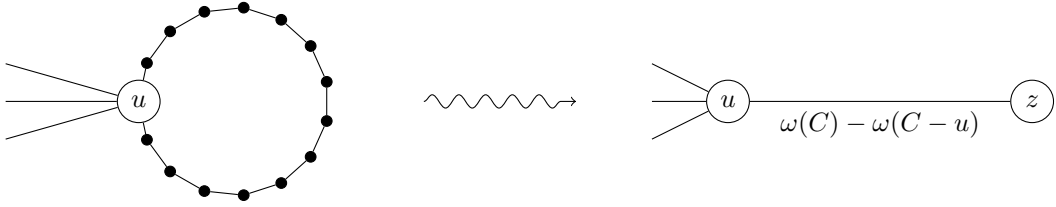
\begin{figure}[t]
	\centering
	\begin{tikzpicture}[auto]
		\def\n{15}
		\def\radius{1.25}
		\def\firstAngle{180}
		\node[alter] (K-1) at ({\radius * cos(\firstAngle))},{\radius * sin(\firstAngle))}) {$u$};
		\foreach \i in {2,...,\n} {
			\node[knoten] (K-\i) at ({\radius * cos(\firstAngle + 360 * \i / \n - 360 / \n))},{\radius * sin(\firstAngle + 360 * \i / \n - 360 / \n))}) {};
			\pgfmathtruncatemacro{\j}{\i - 1};
			\path (K-\j) edge[-] (K-\i);
		}
		\path (K-1) edge[-] (K-\n);
		
		\draw[majarr] (K-1) edge (-3,-0.5);
		\draw[majarr] (K-1) edge (-3,0);
		\draw[majarr] (K-1) edge (-3,0.5);

		\begin{scope}[xshift=6.5cm]
			\draw [->,snake=snake,line after snake=1mm] (-4,0) -- (-2,0);
			\node[alter] at (0,0) (u) {$u$};
			\node[alter] at (4,0) (z) {$z$};

			\draw[majarr] (u) edge node[swap]{$\omega(C) - \omega(C-u)$} (z);
			\draw[majarr] (u) edge (-1,-0.5);
			\draw[majarr] (u) edge (-1,0);
			\draw[majarr] (u) edge (-1,0.5);
		\end{scope}
	\end{tikzpicture}
	\caption{Left: A \pendcycle{}~$C$ with~$u$ being the vertex of degree at least three. Right: The graph after applying \cref{rule:pend-cycle} where~$s$ is decreased by~$\omega(C-u)$.}
	\label{fig:rrule cycle3}
\end{figure}
Then, matching~$z$ corresponds to taking the matching in~$C$ and not matching~$z$ corresponds to taking the matching in~$C-u$.
Formalizing this idea, we arrive at the following data reduction rule.

\begin{rrule}
\label{rule:pend-cycle}
	Let $G$ be a graph with non-negative edge weights.
	Let~$C$ be a \pendcycle{} in~$G$, where $u \in C$ has degree at least three in $G$. 
	Then replace $C$ by an edge $uz$ with $\omega(uz)=\omega(C) - \omega(C-u)$ and 
	decrease the solution value~$s$ by $\omega(C-u)$, where $z$ is a new vertex.
\end{rrule}

\begin{lemma}
\label{thm:cycle correct}
	\cref{rule:pend-cycle} is correct.	
\end{lemma}
\begin{proof}
	Let~$C$ be a \pendcycle{} in~$G$ where~$u \in C$ has degree at least three in~$G$ and let~$G'$ be the graph obtained by applying \cref{rule:pend-cycle} to~$C$.
	We show~$\omega(G') = \omega(G) - \omega(C-u)$. 

	Let~$M$ be a maximum-weight matching in~$G$. 
	Let~$M_{\overline{C}} := M \setminus E(C)$.
	Observe that~$\omega(M_{\overline{C}}) = \omega(M) - \omega(M \cap E(C)) \ge \omega(G) - \omega(C)$.
	If~$u$ is matched with respect to~$M_{\overline{C}}$, then we have~$M_{\overline{C}} = M \setminus E(C-u)$.
	Hence, $\omega(G') \ge \omega(M_{\overline{C}}) \ge \omega(G) - \omega(C-u)$.
	If~$u$ is free with respect to $M_{\overline{C}}$, then~$M_{\overline{C}} \cup \{uz\}$ is a matching in $G'$ with weight at least $(\omega(G)-\omega(C))+(\omega(C)-\omega(C-u))=\omega(G)-\omega(C-u)$.
	Hence, in both cases we have~$\omega(G') \ge \omega(G) - \omega(C-u)$.
	
	Conversely, let $M'$ be a maximum-weight matching in $G'$.
	Recall that, for an edge-weighted graph~$H$, $\match(H)$ denotes a maximum-weight matching in~$H$.
	If~$uz \in M'$, then $(M' \setminus \{uz\}) \cup \match(C)$ is a matching in~$G$ with~$\omega(G') - (\omega(C) - \omega(C-u)) + \omega(C) = \omega(G')+\omega(C-u)$.
	Hence, $\omega(G) \ge \omega(G') + \omega(C-u)$.
	If~$uz \not\in M'$, then~$M' \cup \match(C-u)$ is a matching in~$G$ with weight at least $\omega(G') + \omega(C-u)$.
	Again, in both cases we have~$\omega(G) \ge \omega(G') + \omega(C-u)$.
	Combined with~$\omega(G') \ge \omega(G) - \omega(C-u)$, we arrive at~$\omega(G') = \omega(G) - \omega(C-u)$.
\end{proof}

The basic idea for \maxpaths{} is the same as for \pendcycles{}. 
The difference is that we have to distinguish four cases depending on whether or not the two endpoints~$u$ and~$v$ of a \maxpath{}~$P$ are matched within~$P$.
To avoid some trivial case distinctions, we assume that~$\omega(uv) = 0$ if the edge~$uv$ does not exist in~$G$.
We denote by $P-u-v$ the path obtained from removing in~$P$ the vertices $u$ and~$v$.

\Cref{fig:rrule-deg2} visualizes the next data reduction rule. 

\begin{rrule}
\label{rule:max-path}
	Let~$G = (V,E)$ be a graph with non-negative edge weights~$\omega\colon E \rightarrow \N$.
	Let~$P$ be a \maxpath{} in~$G$ with endpoints~$u$ and~$v$. 
	Then remove all vertices in~$P$ except~$u$ and $v$, 
	then add a new vertex~$z$ and, if not already existing, add the edge~$uv$.
	Furthermore, set~$\omega(uz):=\omega(P-v) - \omega(P-u-v)$, $\omega(vz):=\omega(P-u) - \omega(P-u-v)$, and~$\omega(uv):=\max\{\omega(uv),\omega(P) - \omega(P-u-v)\}$, and decrease the solution value~$s$ by~$\omega(P-u-v)$. 
\end{rrule}

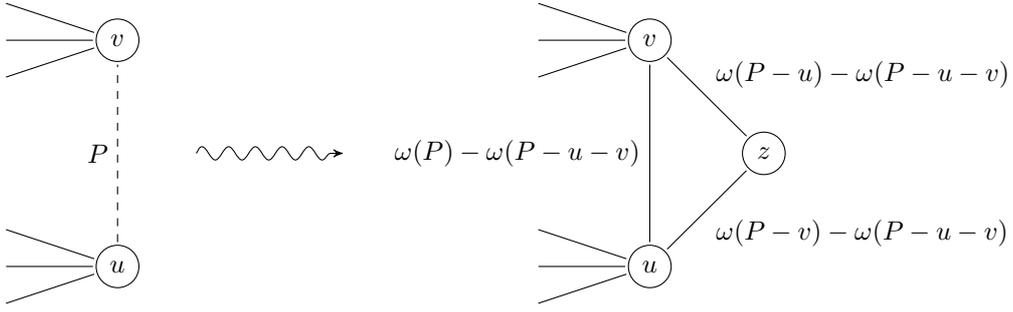
\begin{figure}[t]
	\centering
	\begin{tikzpicture}[auto, >=stealth',shorten <=1pt, shorten >=1pt]
		\node[alter] at (0,0) (u) {$u$};
		\node[alter] at (0,3) (v) {$v$};

		\draw[majarr,dashed] (u) edge node{$P$} (v);
		
		\draw[majarr] (u) edge (-1.5,-0.5);
		\draw[majarr] (u) edge (-1.5,0);
		\draw[majarr] (u) edge (-1.5,0.5);
		\draw[majarr] (v) edge (-1.5,2.5);
		\draw[majarr] (v) edge (-1.5,3);
		\draw[majarr] (v) edge (-1.5,3.5);

		\begin{scope}[xshift=7cm]
			\draw [->,snake=snake,line after snake=1mm] (-6,1.5) -- (-4,1.5);

			\node[alter] at (0,0) (u) {$u$};
			\node[alter] at (0,3) (v) {$v$};
			\node[alter] at (1.5,1.5) (z) {$z$};

			\draw[majarr] (u) edge (-1.5,-0.5);
			\draw[majarr] (u) edge (-1.5,0);
			\draw[majarr] (u) edge (-1.5,0.5);
			\draw[majarr] (v) edge (-1.5,2.5);
			\draw[majarr] (v) edge (-1.5,3);
			\draw[majarr] (v) edge (-1.5,3.5);
			
			\draw[majarr] (u) edge node[swap]{$\omega(P-v) - \omega(P-u-v)$} (z);
			\draw[majarr] (v) edge node{$\omega(P-u) - \omega(P-u-v)$} (z);
			\draw[majarr] (u) edge node{$\omega(P) - \omega(P-u-v)$} (v);
		\end{scope}
	\end{tikzpicture}
	\caption{Applying \cref{rule:max-path} on a path~$P$ with endpoints~$u$ and~$v$ (where~$u$ and~$v$ are not adjacent). The four choices for~$u$ and~$v$ on whether or not they are matched to a vertex within the path are reflected by the three (full) edges on the right where at most one can be taken into a matching. Since the edge~$uv$ is not contained in the input graph the weight of the edge~$uv$ in the reduced graph simplifies to the displayed value.
	}
	\label{fig:rrule-deg2}
\end{figure}

\begin{lemma}
\label[lemma]{thm:path correct}
	\cref{rule:max-path} is correct.
\end{lemma}
\begin{proof}
	Let~$G$ be the input graph with a \maxpath{}~$P$ with endpoints~$u$ and~$v$.
	Furthermore, let~$G'$ be the reduced instance with~$z$ defined as in the data reduction rule.
	We show that~$\omega(G') = \omega(G) - \omega(P-u-v)$.
	
	Let~$M$ be a maximum-weight matching for~$G$. 
	We define~$M_{\overline{P}} := M \setminus E(P)$.
	Observe that~$\omega(M_{\overline{P}}) = \omega(M) - \omega(M \cap E(P)) \ge \omega(G) - \omega(P)$.
	We consider four cases.
	\begin{enumerate}
		\item If both~$u$ and~$v$ are matched with respect to~$M_{\overline{P}}$, then~$M_{\overline{P}} = M \setminus E(P - u - v)$ and hence
				\begin{align}
						\omega(M_{\overline{P}}) = \omega(M) - \omega(M \cap E(P - u - v)) \ge \omega(G) - \omega(P-u-v).
				\end{align}

		\item Let one vertex in $\{u,v\}$ be matched and let one be free.
				Without loss of generality, we assume that~$u$ is matched and $v$ is free with respect to $M_{\overline{P}}$.
				Then, we have that $M_{\overline{P}} = M \setminus E(P - u)$ and hence~$\omega(M_{\overline{P}}) \ge \omega(G) - \omega(P-u)$.
				Thus, $M_{\overline{P}} \cup \{vz\}$ is a matching of weight at least
				\begin{align*}
						(\omega(G) - \omega(P-u)) + (\omega(P-u)-\omega(P-u-v)) = \omega(G) - \omega(P-u-v).
				\end{align*}
		\item Finally, if both~$u$ and~$v$ are free with respect to~$M_{\overline{P}}$, then~$M_{\overline{P}} \cup \{uv\}$ is a matching of weight at least $(\omega(G) - \omega(P)) + (\omega(P)-\omega(P-u-v)) = \omega(G) - \omega(P-u-v)$.
	\end{enumerate}
	Thus in each case we have~$\omega(G') \ge \omega(G) - \omega(P-u-v)$.
	
	Conversely, let $M'$ be a maximum-weight matching for~$G'$.
	We define $\overline{M'} := M' \setminus \{uz, vz, uv \}$.
	Again, we distinguish four cases.
	\begin{enumerate}
		\item If both~$u$ and~$v$ are matched with respect to~$\overline{M'}$, then $\overline{M'}=M'$.
			Hence, $\overline{M'} \cup \match(P-u-v)$ is a matching in~$G$ with weight at least~$\omega(G') + \omega(P-u-v)$.
		\item If~$u$ is matched and~$v$ is free with respect to~$\overline{M'}$, then w.l.o.g.~$vz \in M'$.
			Hence, $\overline{M'} \cup \match(P-u)$ is a matching in $G$ with weight at least~$\omega(G') - (\omega(P-u) - \omega(P-u-v)) + \omega(P-u) = \omega(G') + \omega(P-u-v)$.
		\item If~$u$ is matched and~$v$ is free with respect to~$\overline{M'}$, then w.l.o.g.~$uz \in M'$.
			Hence, $\overline{M'} \cup \match(P-v)$ is a matching in $G$ with weight at least~$\omega(G') - (\omega(P-v) - \omega(P-u-v)) + \omega(P-v) = \omega(G') + \omega(P-u-v)$.
		\item Finally, if both~$u$ and~$v$ are free with respect to~$\overline{M'}$, then w.l.o.g~$uv \in M'$ as~$\omega(uv) \ge \omega(uz)$ and~$\omega(uv) \ge \omega(vz)$.
			Now, we encounter two subcases.
			\begin{enumerate}
				\item If~$\omega(uv) > \omega(P) - \omega(P-u-v)$, then the edge~$uv$ is in~$G$ and in~$G'$, having the same weight in both graphs.
				Then, $M' \cup \match(P-u-v)$ is a matching in~$G$ with weight at least~$\omega(G') + \omega(P-u-v)$.
				\item Otherwise, $\overline{M'} \cup \match(P)$ is a matching in~$G$ with weight at least~$\omega(G') - (\omega(P) - \omega(P-u-v)) + \omega(P) = \omega(G') + \omega(P-u-v)$.
			\end{enumerate}
	\end{enumerate}
	Hence, in all cases we have~$\omega(G) \ge \omega(G') + \omega(P-u-v)$.
	Combined with~$\omega(G') \ge \omega(G) + \omega(P-u-v)$, we can infer that~$\omega(G') = \omega(G) - \omega(P-u-v)$.
\end{proof}

\begin{lemma}
\label[lemma]{lem:deg-two-lin-time-weighted}
	\cref{rule:max-path,rule:pend-cycle} can be exhaustively applied in $O(n+m)$ time.
\end{lemma}
\begin{proof} 
	First, we collect in~$O(n+m)$ time all \maxpaths{} and all \pendcycles{} \cite[Lemma 3]{BDKNN20}.
	Given a \maxpath{} or a \pendcycle{} on~$\ell$ vertices due to \cref{obs:path-cycle-lin-time} one can compute the necessary maximum-weight matchings (at most four) in~$O(\ell)$ time.
	Moreover, replacing the \maxpath{} or the \pendcycle{} by the respective structure is doable in~$O(\ell)$ time.
	Applying \cref{rule:max-path,rule:pend-cycle} does not create new \maxpaths{} (recall that a \maxpath{} needs at least two vertices of degree two) or \pendcycles{}.
	Thus, as all \maxpaths{} and \pendcycles{} combined contain at most~$n$ vertices, \cref{rule:max-path,rule:pend-cycle} can be exhaustively applied in $O(n+m)$ time.
\end{proof}

Each of \cref{rule:deg-zero-weighted,rule:deg-one-weighted,rule:max-path} can be exhaustively applied in linear time; however, we do not know whether all these data reduction rules together can be exhaustively applied in linear time.
Note that after applying \cref{rule:pend-cycle} \cref{rule:deg-one-weighted} might become applicable. 
For our problem kernel below it suffices to apply \cref{rule:deg-zero-weighted,rule:deg-one-weighted,rule:max-path} 
in a specific order (using \cref{lem:deg-one-weighted-time,lem:deg-two-lin-time-weighted}).
Note that we might output a problem kernel where \cref{rule:deg-zero-weighted,rule:deg-one-weighted} are applicable.
In our experimental part it turned out that it is beneficial to apply the rules exhaustively (in superlinear time)
to reduce the input graph as much as possible.

\begin{theorem}
	\label{thm:w-kernel}
	\WMatch admits a linear-time computable $13k$-vertex and $17k$-edge  kernel with respect to the parameter feedback edge number~$k$. %
\end{theorem}
\begin{proof}
	Let~$G = (V,E)$ be the input instance and $F \subseteq E$ a feedback edge set of size at most $k$.
	Without loss of generality,
	one can assume that the input graph does not contain a cycle 
	where each vertex has degree two, 
	or a path where the endpoints have degree one and the internal vertices have degree two, 
	because otherwise such a cycle or path can be solved independently 
	in linear time (see~\cref{obs:path-cycle-lin-time}).

	The kernelization algorithm works as follows: 
	First, exhaustively apply \cref{rule:deg-zero-weighted} in~$O(n+m)$ time.
	Second, exhaustively apply \cref{rule:deg-one-weighted} in $O(n+m)$ time (see~\cref{lem:deg-one-weighted-time}). 
	Third, exhaustively apply \cref{rule:pend-cycle,rule:max-path} in $O(n+m)$ time (see \cref{lem:deg-two-lin-time-weighted}).
	Note that when applying the rules in this order, 
	the resulting graph~$\widehat G = (\widehat V, \widehat E)$ does not contain any \maxpaths, or \pendcycles.
	But for each \pendcycle{} we introduced a new degree one vertex.
	However, for each pending cycle there is at least one (distinct) edge in $F$.
	Hence, $\widehat G$ has at most $k$ degree one vertices.
	Let $V_1$ be the set of degree-one vertices in $\widehat G$.
	Moreover, after the second step of our kernelization (\cref{rule:deg-one-weighted}) the graph contains at most~$3k$ \maxpaths{}~\cite[Lemma 2]{BDKNN20}.
	Thus, a feedback edge set~$\widehat F \subseteq \widehat E$ for~$\widehat G$ of minimum size
	contains at most $4k$ edges (each application of \cref{rule:max-path} increases the feedback edge set by one).

	To analyze the size of~$\widehat G$ in terms of $k$, 
	we transform $\widehat G$ into a forest by cutting the edges in~$\widehat F$ such that 
	for each edge in~$\widehat F$ two new vertices of degree one are introduced.
	Formally, we have graph~$G'=(V',E')$, 
	where~$V' := \widehat V \cup \left\{ v_u^{uw},v_w^{uw} \mid uw \in \widehat F \right\}$ 
	and $E' := (\widehat E \setminus \widehat F) \cup \left\{ uv_u^{uw},wv_w^{uw} \mid  uw \in \widehat F \right\}$.

	Observe that $G'$ is a forest where $(V' \setminus \widehat V) \cup V_1$ are the leaves 
	and~$\widehat V \setminus V_1$ are the internal vertices.
	Hence, we have at most $9k$ leaves and thus at most $9k-1$ internal vertices of degree at least three.
	Since there are at most $3k$ internal vertices of degree two (one for each application of \cref{rule:max-path}), 
	we have $|\widehat V \setminus V_1| < 12k$ and $|\widehat V | \leq 13k$.
	Furthermore, $\widehat E \setminus \widehat F$ are edges 
	of the forest $G'[\widehat V]$.
	Hence, we have~$|\widehat E| = |\widehat E \setminus \widehat F| + |\widehat F| < 17k$.
\end{proof}

\section{Experimental Evaluation}
\label{sec:experiments}

In this section, we provide an experimental evaluation of the presented 
data reduction rules on real-world graphs ranging from a few thousand vertices and edges 
to a few million vertices and edges.
We analyze the effectiveness and efficiency of the kernelization as well as the effect on the subsequently used state-of-the-art solvers of \citet{HSt17,KP98}, and \citet{Kol09}.

In \cref{sec:setup}, we give details about our test scenario.
Then we first focus in \cref{sec:eval-kernel} on the evaluation of \cref{rule:crown,rule:deg-zero-one-vertices,rule:deg-two-vertices,rule:deg-one-weighted,rule:deg-zero-weighted,rule:pend-cycle,rule:max-path} in terms of running time needed to apply them and size of resulting instances.
In \cref{sec:eval-time-unweighted}, we then analyze the effect of applying \cref{rule:crown,rule:deg-zero-one-vertices,rule:deg-two-vertices} in combination with a solver for \CMatch.
Afterwards, in \cref{sec:eval-time-weighted}, we analyze the effect of applying \cref{rule:deg-one-weighted,rule:deg-zero-weighted,rule:pend-cycle,rule:max-path} in combination with a solver for \WMatch.

\subsection{Setup and Implementation Details}
\label{sec:setup}
Our program is written in C++14 and the source code is available from \url{https://git.tu-berlin.de/akt-public/matching-data-reductions.git}.
One can replicate all experiments by following the manual provided with the source code.
We ran all our experiments on an Intel(R) Xeon(R) CPU E5-1620 3.60\,GHz machine with 64\,GB main memory 
under the Debian GNU/Linux 7.0 operating system, 
where we compiled the program (including the solvers of \cite{Kol09,KP98}) with GCC~7.3.0.
For the solver of \citet{HSt17} we used Python 2.7.15rc1.

\paragraph{Data set.}
All tested graphs are from the established SNAP~\cite{snap} data set with a time limit of one hour per instance.
See \Cref{tab:short-list} for a sample list of graphs with their respective numbers of vertices and edges. 
\begin{table}
\def\ModRows{8}
	\caption{A selection of our test graphs from SNAP~\cite{snap} with their respective size.}
	\label{tab:short-list}
	\begin{center}
	\pgfplotstabletypeset[
	columns={filename,misc_n,misc_m},
		columns/filename/.style={string type,column name=Graph,column type = {r}},
		columns/misc_n/.style={column name=$n$,precision=1,column type = {r}},
		columns/misc_m/.style={column name=$m$,precision=1,column type = {r}},
		every head row/.style ={before row=\toprule, after row=\midrule},
		every last row/.style ={after row=\bottomrule},
		row predicate/.code={%
			\pgfmathparse{int(mod(#1,\ModRows)}%
			\ifnum\pgfmathresult=1\relax%
			\else\pgfplotstableuserowfalse%
			\fi%
	}, col sep = comma] {\resultsAllTab}
	\pgfplotstabletypeset[
	columns={filename,misc_n,misc_m},
		columns/filename/.style={string type,column name=Graph,column type = {r}},
		columns/misc_n/.style={column name=$n$,precision=1,column type = {r}},
		columns/misc_m/.style={column name=$m$,precision=1,column type = {r}},
		every head row/.style ={before row=\toprule, after row=\midrule},
		every last row/.style ={after row=\bottomrule},
		row predicate/.code={%
			\pgfmathparse{int(mod(#1,\ModRows)}%
			\ifnum\pgfmathresult=3\relax%
			\else\pgfplotstableuserowfalse%
			\fi%
	}, col sep = comma] {\resultsAllTab}
	\end{center}
\end{table}
The full list is given in \cref{tab:unweighted} in the Appendix.
The weighted graphs are generated from the unweighted graphs by adding edge-weights between~1 and~1000 chosen independently and uniformly at random.

\paragraph{Implementation details of our kernelization algorithms.}
We implemented kernelization algorithms for the unweighted and weighted case.
The first kernelization is for \CMatch{}, which exhaustively applies \cref{rule:deg-zero-one-vertices,rule:deg-two-vertices}.
Our implementation here is rather simplistic in the sense that it maintains a list of degree-one and degree-two vertices which are processed one after the other (in a straightforward manner). 
We apply the rules exhaustively although this gives in theory a super-linear running time.
Note that one can (theoretically) improve our implementation of \cref{rule:deg-two-vertices} by a linear-time algorithm of \citet{BK09a}.
Very recently, \citet{KLPU20} provided a fine-tuned algorithm that exhaustively applies \cref{rule:deg-zero-one-vertices,rule:deg-two-vertices} on bipartite graphs roughly three times faster than our naive implementation; their general approach should be also applicable for general graphs.
However, our naive (super-linear time) implementation for exhaustively applying \cref{rule:deg-zero-one-vertices,rule:deg-two-vertices} was at least two times faster than reading and parsing the input graph and at least three times faster than the fastest implementation for finding maximum-cardinality matchings.
Thus, applying \cref{rule:deg-zero-one-vertices,rule:deg-two-vertices} was not a bottleneck in our implementation and we did not optimize it further.

The second kernelization is also \CMatch{} and it exhaustively applies \cref{rule:crown}.
To this end, we used the algorithm described by \citet{IOY14}.
The main steps of this algorithm are:
\begin{enumerate}
	\item compute a maximum-cardinality matching in a given bipartite graph~$\overline{G}$ (to compute the initial LP-solution; see \cref{ssec:crowns-LP}) and \label{step:bip-match} 
	\item determine the topological ordering of the DAG formed by the strongly connected components of a given digraph~$\overline{D}$. \label{step:topOr-DAG-SCC}
	
		(The underlying undirected graph of~$\overline{D}$ is~$\overline{G}$; the matching computed in Step~\ref{step:bip-match} determines how the edges in~$\overline{G}$ are directed in~$\overline{D}$. 
		Each crown in the input graph~$G$ corresponds to a strongly connected component in~$\overline{D}$; refer to \citet{IOY14} for details.)
\end{enumerate}
Both steps can be solved using classic algorithms.
We implemented for Step~\ref{step:bip-match} the classic $O(\sqrt{n}m)$-time algorithm of \citet{HK73} for finding a bipartite matching in~$\overline{G}$.
For Step \ref{step:topOr-DAG-SCC}, we implemented Kosaraju's algorithm~\cite{AHU83} for finding the strongly connected components of~$\overline{D}$ in reverse topological order.

The third kernelization is for \WMatch{}. 
We use the algorithms described in \cref{lem:deg-one-weighted-time,lem:deg-two-lin-time-weighted} to apply \cref{rule:pend-cycle,rule:max-path,rule:deg-one-weighted}.
Deviating from the algorithm described in \cref{thm:w-kernel}, based on empirical observations our program applies \cref{rule:pend-cycle,rule:max-path,rule:deg-one-weighted,rule:deg-zero-weighted} as long as possible.
Hence, the kernelization does not run in linear time but further shrinks the input graph.

\newcommand{\solverKPE}{\texttt{KP-Edm}\xspace}
\newcommand{\solverKE}{\texttt{Kol-Edm}\xspace}
\newcommand{\solverKEW}{\texttt{Kol-Edm-W}\xspace}
\newcommand{\solverHSMV}{\texttt{HS-MV}\xspace}

\begin{table}
	\caption{Set of solvers we used in our experiments.
		Here, ``\textsc{MM}$\leadsto$\textsc{W-PM}'' is the folklore reduction from \CMatch to
		\textsc{Minimum Weighted Perfect Matching} and ``\textsc{W-M}$\leadsto$\textsc{W-PM}'' is
		the folklore reduction from \WMatch to \textsc{Minimum Weighted Perfect Matching}.
	}
	\centering
	\begin{tabular}{l r r r}  \toprule
		
		acronym 		& implementation by	& core algorithm 										& language \\ \midrule
		\solverKPE&\citet{KP98}&\citet{Edm65} 										& C\\
		\solverKE& \citet{Kol09} 	& \citet{Edm65,Edm65-2},\textsc{MM}$\leadsto$\textsc{W-PM} 	& C++\\
		\solverKEW& \citet{Kol09} 	& \citet{Edm65,Edm65-2},\textsc{W-M}$\leadsto$\textsc{W-PM}	& C++\\
		\solverHSMV& \citet{HSt17} 	& \citet{MV80} 										& Python 2.7\\

		\bottomrule
	\end{tabular}
	\label{tab:solvers}
\end{table}

\paragraph{Used solvers.}
To test the effect of our data reduction rules, we compare the running time of a solver on an input instance 
against the running time of our kernelization procedure plus the running time of the same solver on the output 
of our kernelization procedure.
We refer to \cref{tab:solvers} for an overview of the tested solvers.
For the data reductions rules for \WMatch we used the solver of \citet{Kol09} (implemented in C++) which is a fine-tuned version 
of Edmonds' algorithm for \textsc{Minimum-Weight Perfect Matching} \cite{Edm65,Edm65-2}. 
Note that the solver of \citet{Kol09} finds perfect matchings of minimum weight. 
We thus use this solver on graphs obtained from applying the folklore reduction from \WMatch (and thus also from \CMatch) to \textsc{Minimum-Weight Perfect Matching}. 
Applied on a graph~$G$ with~$n$ vertices and~$m$ edges, the reduction adds a copy of~$G$ and adds a weight-zero edge between each vertex in~$G$ and its added copy.
Thus, the resulting graph can be computed in linear time and has~$2n$ vertices and~$2m + n$ edges.
To the best of our knowledge, the solver of \citet{Kol09} plus the folklore reduction yields the currently fastest algorithm for \WMatch.
For the rest of this paper, \solverKEW denotes the solver of \citet{Kol09} plus the folklore reduction for \WMatch.

To test the data reduction rules for \CMatch, we used three different solvers.
First, we used the solver (denoted by \solverKE) of \citet{Kol09} plus the folklore reduction for \CMatch, 
which we get basically for free from the weighted case.
Second, we used the solver (denoted by \solverKPE) of \citet{KP98} (implemented in C) which is a fine-tuned version of Edmonds' algorithm \cite{Edm65}.
To the best of our knowledge this is in practice still the fastest algorithm.
In our experiments, \solverKPE was clearly the fastest solver.
Third, we used the solver (denoted by \solverHSMV) of \citet{HSt17} (implemented in Python).\footnote{In a few cases \solverHSMV returned an edge set that is not a matching. However, a maximum-cardinality matching was easily recoverable from the returned edge set by removing one or two edges. The authors (Huang and Stein) and are working on a fix for this issue.}
This is the only implementation of the Micali-Vazirani algorithm \cite{MV80} we are aware of.
The Micali-Vazirani algorithm has currently the best asymptotic worst-case running time. 
However, in our tests \solverHSMV (Python) was clearly outperformed by \solverKPE (C).%

\subsection{Efficiency and Effectiveness of our Data Reduction Rules}
\label{sec:eval-kernel}

\paragraph{Effectiveness of our rules.}
The effectiveness of our kernelization algorithms is displayed in \Cref{fig:kernel-size2}: 
\begin{figure}[t!]
	\begin{tikzpicture}[scale=1.0]
		\begin{axis}[
				width=\textwidth,
				height=0.42\textwidth,
				ylabel={\% [$100 \% = n + m$]},
				legend style = {
						at={(1, 1.02)},
						anchor={south east},
						font = \small
				},
				grid,
				legend cell align = left,
				legend columns = 5,
				ytick={0,25,50,75,100,200},
				xtick=data,
				xticklabels from table={\namesTable}{filename}
				,tick label style={font=\footnotesize}  
				,xmax=2+\TotalRowsResultsAll
				,xmin=0
				,x tick label style={rotate=60,anchor=east}
				]

			\addplot[blue,mark=triangle] table 
				[y expr=100*((\thisrow{k_crown_result_n_part_2}+\thisrow{k_crown_result_m_part_2}) / (\thisrow{misc_n} + \thisrow{misc_m})), x={id},col sep = comma] {\resultsAllTab};
			\addlegendentry{\Cref{rule:crown,rule:deg-zero-one-vertices,rule:deg-two-vertices}}

			\addplot[red,mark=o] table 
				[y expr=100*((\thisrow{k2_n_part_2}+\thisrow{k2_m_part_2}) / (\thisrow{misc_n} + \thisrow{misc_m})), x={id},col sep = comma] {\resultsAllTab};
			\addlegendentry{\Cref{rule:crown}}

			\addplot[blue!50!red,mark=triangle*] table 
				[y expr=((\thisrowno{7} + \thisrowno{8}) / (\thisrowno{1} + \thisrowno{2})) * 100, x={id}] {\dataUnweightedPerm};
			\addlegendentry{\Cref{rule:deg-zero-one-vertices,rule:deg-two-vertices}}

		\end{axis}
	\end{tikzpicture}
	\begin{tikzpicture}[scale=1.0]
		\begin{axis}[
				width=\textwidth,
				height=0.42\textwidth,
				ylabel={\% [$100 \% = n + m$]},
				legend style = {
						at={(1, 1.02)},
						anchor={south east},
						font = \small
				},
				grid,
				legend cell align = left,
				legend columns = 5,
				ytick={0,25,50,75,100,200},
				xtick=data,
				xticklabels from table={\namesTable}{filename}
				,tick label style={font=\footnotesize}  
				,xmax=2+\TotalRowsResultsAll
				,xmin=0
				,x tick label style={rotate=60,anchor=east}
				]

			\addplot[blue,mark=triangle] table 
				[y expr=100*((\thisrow{k_crown_result_n_part_2}+\thisrow{k_crown_result_m_part_2}) / (\thisrow{misc_n} + \thisrow{misc_m})), x={id},col sep = comma] {\resultsAllTab};
			\addlegendentry{\Cref{rule:crown,rule:deg-zero-one-vertices,rule:deg-two-vertices}}

			\addplot[orange,mark=square,discard if={kerneltime}{}] table 
				[x={id}, y expr=(100 * (\thisrow{kernelvertices}+\thisrow{kerneledges}) / (\thisrow{vertices}+\thisrow{edges})), col sep = comma]{\dataHighWeighted};
			\addlegendentry{\Cref{rule:deg-one-weighted,rule:deg-zero-weighted,rule:pend-cycle,rule:max-path}}

			\addplot[black,mark=x] table 
				[y expr=((\thisrow{k1_rule_0}+\thisrow{k1_rule_1}) / (\thisrow{k1_rule_0}+\thisrow{k1_rule_1}+\thisrow{k1_rule_2})) * 100, x={id}, col sep = comma] {\resultsAllTab};
			\addlegendentry{\# of app.\@ \Cref{rule:deg-zero-one-vertices}}
		\end{axis}
	\end{tikzpicture}
	\vspace{-10pt}
	\caption{Kernel sizes (in \%; 100\,\% = $n+m$ of input graph) for several subsets of our data reduction rules.
	We tested: all unweighted rules (\Cref{rule:crown,rule:deg-zero-one-vertices,rule:deg-two-vertices}), only the crown rule (\cref{rule:crown}), only the unweighted rules for low degree vertices (\Cref{rule:deg-zero-one-vertices,rule:deg-two-vertices}), and all weighted rules (\Cref{rule:deg-one-weighted,rule:deg-zero-weighted,rule:pend-cycle,rule:max-path}).
	The crosses show the number of applications of \cref{rule:deg-zero-one-vertices} (in \%; 100\,\% = number of applications of \Cref{rule:deg-zero-one-vertices,rule:deg-two-vertices}).  
	The graphs are ordered in both plots by relative size of the kernel after applying \cref{rule:crown,rule:deg-zero-one-vertices,rule:deg-two-vertices}.
	}	
	\label{fig:kernel-size2}
\end{figure}
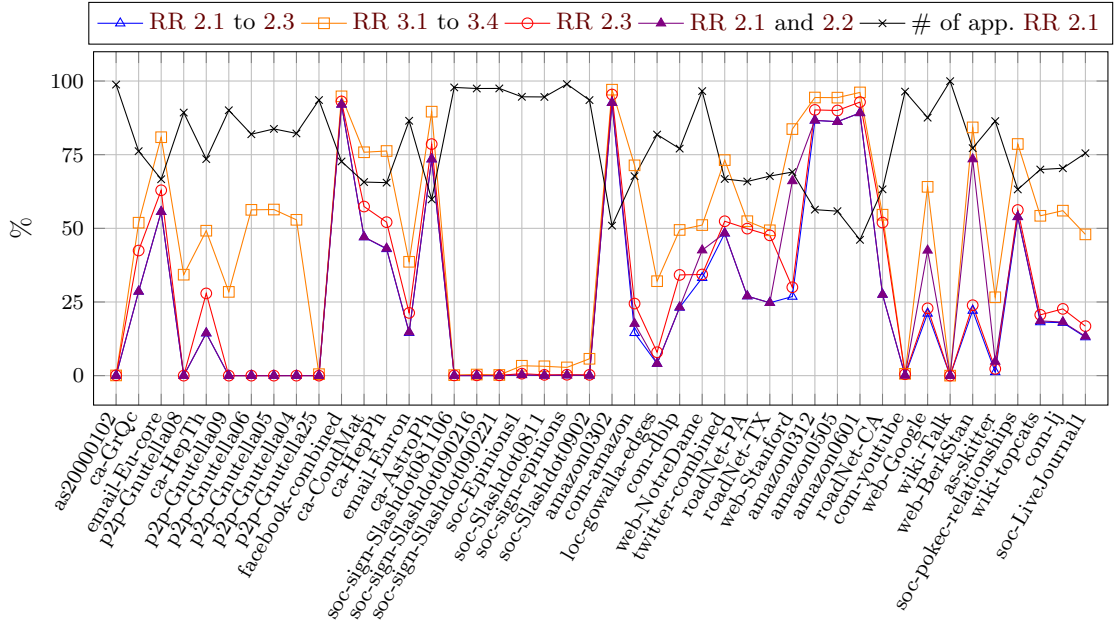
Few graphs remained almost unchanged while other graphs were essentially solved by the kernelization algorithm.

For the unweighted case the situation is as follows:
On the~44 tested graphs, on average 72\% of the vertices and edges are removed by the kernelization; the median is 82\%.
The least amenable graph was \texttt{amazon0302} with a size reduction of only 7\%. 
In contrast, on 16 out of the~44~graphs the kernelization algorithm reduces more than 99\% of the vertices and edges.
This is mostly due to \cref{rule:deg-zero-one-vertices,rule:deg-two-vertices}: 
if \cref{rule:deg-zero-one-vertices,rule:deg-two-vertices} are exhaustively applied,
then an application of the crown rule (\cref{rule:crown}) further
reduces the kernel size only on four instance substantially.
Moreover, a closer look on how often the 
degree-based rules \cref{rule:deg-zero-one-vertices,rule:deg-two-vertices} are applied
reveals that on the majority of our tested graphs \cref{rule:deg-zero-one-vertices} 
is applied twice as much as \cref{rule:deg-two-vertices}.

While the data reduction rules are less effective in the weighted case (see \cref{fig:kernel-size2}), 
they reduce the graphs on average still by 51\% with the median value being a bit lower with 48\%.
The least amenable graph is again \texttt{amazon0302} with a size reduction of only 3\%. 

\paragraph{Efficiency of our rules.}
In \cref{fig:crown-is-slow}, we compare the running times of the \solverKPE solver (the fastest state-of-the-art solver on our instances for the unweighted case) when applied directly on the input graph together with running times of our data reduction rules for the unweighted case (\Cref{rule:deg-zero-one-vertices,rule:deg-two-vertices,rule:crown}).

\pgfplotstablesort[sort key=k_degree_result_time]{\sortedTable}{\namesTable}

\begin{figure}[t!]
	\begin{tikzpicture}[scale=1.0]
		\begin{axis}[
				width=\textwidth,
				height=0.45\textwidth,
				ymode=log,
				ylabel={time in seconds},
				legend style = {
						at={(1, 1.02)},
						anchor={south east},
						font = \small
				},
				grid,
				legend cell align = left,
				legend columns = 6,
				ytick={0.001,0.01,0.1,1,10,100,1000,10000,100000,1000000},
				xtick=data,
				xticklabels from table={\sortedTable}{filename}
				,tick label style={font=\footnotesize}  
				,xmax=2+\TotalRowsResultsAll
				,xmin=0
				,x tick label style={rotate=60,anchor=east}
				]
			\addplot[blue,mark=triangle] table [y expr=(\thisrow{k_degree_result_time} + \thisrow{k_crown_result_lp_time} + \thisrow{k_crown_result_time} + \zeroOffset)/1000, x expr=\coordindex + 1,col sep = comma] {\sortedTable};
			\addlegendentry{\Cref{rule:deg-zero-one-vertices,rule:deg-two-vertices,rule:crown}}

			\addplot[red,mark=o] table [y expr=((\thisrow{k2_lp_time} + \thisrow{k2_time} + \zeroOffset) / 1000), x expr=\coordindex + 1,col sep = comma] {\sortedTable};
			\addlegendentry{\Cref{rule:crown}}

			\addplot[blue!50!red,mark=triangle*] table [y expr=(\thisrow{k_degree_result_time} +  \zeroOffset)/1000, x expr=\coordindex + 1,col sep = comma] {\sortedTable};
			\addlegendentry{\Cref{rule:deg-zero-one-vertices,rule:deg-two-vertices}}

			\addplot[green!50!black,mark=*,only marks,filter if empty={c_time}{\sortedTable}
				] table [y expr=(\thisrow{c_time} + \zeroOffset)/1000, x expr=\coordindex + 1] {\sortedTable};
			\addlegendentry{\solverKPE}

		\end{axis}
	\end{tikzpicture}
	\vspace{-10pt}
	\caption{Running time of various data reduction rules and the solvers (unweighted case; without kernelization). To all values \zeroOffset{}\,millisecond was added to display 0-values. The graphs are ordered by running time for \cref{rule:deg-zero-one-vertices,rule:deg-two-vertices}.}
	\label{fig:crown-is-slow}
\end{figure}
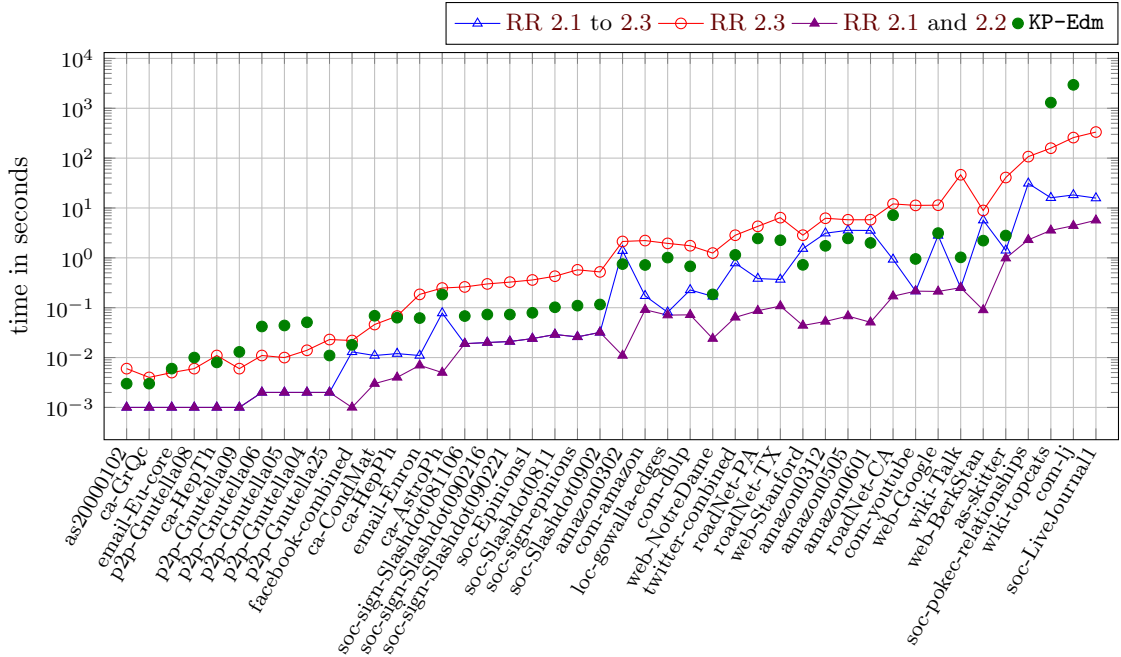

On some graphs it does take more time to apply \cref{rule:crown,rule:deg-zero-one-vertices,rule:deg-two-vertices} than executing the \solverKPE solver directly.
But if only \cref{rule:deg-zero-one-vertices,rule:deg-two-vertices} are applied, then the running time of the kernelization stays far below the running time of the \solverKPE solver while the resulting kernel size is mostly the same (see \cref{fig:kernel-size2}).
Applying \cref{rule:crown} without \cref{rule:deg-zero-one-vertices,rule:deg-two-vertices} is clearly not a good idea as of \solverKPE is faster in finding a maximum-cardinality matching.

In \cref{fig:weighted-times}, we compare the running times of the \solverKEW solver (the state-of-the-art solver for the weighted case) when applied directly on the input graph together with running times of our data reduction rules for the weighted case (\Cref{rule:max-path,rule:pend-cycle,rule:deg-one-weighted,rule:deg-zero-weighted}).

\pgfplotstablesort[sort key=kerneltime,col sep = comma]{\sortedWeightedTable}{\dataHighWeighted}

\begin{figure}[t]
	\begin{tikzpicture}[scale=1.0]
		\begin{axis}[
				width=\textwidth,
				height=0.45\textwidth,
				ymode=log,
				ylabel={time in seconds},
				legend style = {
						at={(1, 1.02)},
						anchor={south east},
						font = \small
				},
				grid,
				legend cell align = left,
				legend columns = 6,
				ytick={0.001,0.01,0.1,1,10,100,1000,10000},
				xtick=data,
				xticklabels from table={\sortedWeightedTable}{graphname}
				,tick label style={font=\footnotesize}  
				,xmax=2+\TotalRowsResultsAll
				,xmin=0
				,x tick label style={rotate=60,anchor=east}
				]

 			\addplot[orange,mark=square,discard if={kerneltime}{}] table 
				[x expr=\coordindex + 1,
				y expr=(\thisrowno{6}+\zeroOffset)*0.001]{\sortedWeightedTable};
			\addlegendentry{\Cref{rule:max-path,rule:pend-cycle,rule:deg-one-weighted,rule:deg-zero-weighted}}

 			\addplot[only marks,red,mark=diamond*,
				filter if empty={mmtime}{\sortedWeightedTable}
				] table 
				[x expr=\coordindex + 1, y expr=(\thisrow{mmtime}+\zeroOffset)*0.001)]{\sortedWeightedTable};
			\addlegendentry{\solverKEW}

		\end{axis}
	\end{tikzpicture}
	\vspace{-10pt}
	\caption{Running time of applying \Cref{rule:max-path,rule:pend-cycle,rule:deg-one-weighted,rule:deg-zero-weighted} and the \solverKEW solver on weighted graphs (without kernelization). To all values \zeroOffset{}\,millisecond was added to display 0-values.
	The graphs are ordered by running time for \cref{rule:max-path,rule:pend-cycle,rule:deg-one-weighted,rule:deg-zero-weighted} (the ordering is slightly different from the ordering in \cref{fig:crown-is-slow}.)}
	\label{fig:weighted-times}
\end{figure}
The picture is similar to the unweighted case: 
On most graphs the data reduction rules are applied much faster than the solver (\solverKEW), but there are a few exceptions.

\paragraph{Advice on which reduction rules to apply.}
If one has to find maximum-cardinality matchings on large real-world graphs, then we advise to always apply \cref{rule:deg-zero-one-vertices,rule:deg-two-vertices} before feeding the graph to a solver.
Whether or not applying \cref{rule:crown}, however, depends on the specific type of real-world data at hand and should be tested on a few test cases: In most of our test cases the benefit paid with the higher running times was rather small.
For the weighted case, we advise to apply our data reduction rules, but maybe invest time in a more efficient implementation of the rules (our implementation of the weighted rules could probably profit from further optimizations).

\paragraph{Kernel size: theory versus practice.}
In theory, we are used to measure the effectiveness of data reduction rules in terms of 
provable upper bounds for the size of resulting graph (the kernel) in a function only depending on \emph{some} parameter.
If \cref{rule:deg-zero-one-vertices,rule:deg-two-vertices} are not applicable, then \cref{thm:fes-lin-kernel} states that a resulting graph has at most~$2k$~vertices and at most~$3k$~edges, where $k$ is the \paramEnv{feedback edge number}.
In \cref{fig:kernel-vs-input-vs-fes}, we measure the gap between the actual size of the kernel and the proven upper bound from \cref{thm:fes-lin-kernel}.
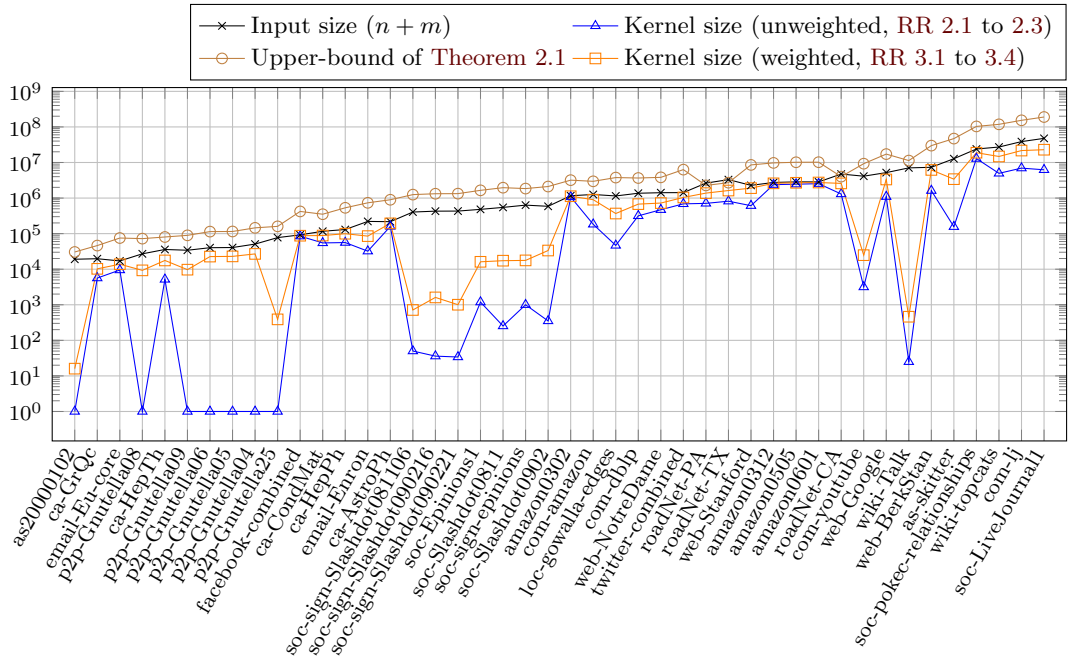
\begin{figure}[t]
	\begin{tikzpicture}[scale=1.0]
		\begin{axis}[
				width=\textwidth,
				height=0.42\textwidth,
				ymode=log,
				legend style = {
						at={(1, 1.02)},
						anchor={south east},
						font = \small
				},
				grid,
				legend cell align = left,
				legend columns = 2,
				ytick={1,10,100,1000,10000,100000,1000000,10000000,100000000,1000000000},
				xtick=data,
				xticklabels from table={\namesTable}{filename}
				,tick label style={font=\footnotesize}  
				,xmax=2+\TotalRowsResultsAll
				,xmin=0
				,x tick label style={rotate=60,anchor=east}
				]
			\addplot[black,mark=x] table [y expr=\thisrow{misc_n} +\thisrow{misc_m} + \zeroOffset, x={id}, col sep = comma] {\resultsAllTab};
			\addlegendentry{Input size ($n+m$)}
			\addplot[blue,mark=triangle] table [y expr=\thisrow{k_crown_result_n_part_2} + \thisrow{k_crown_result_m_part_2} + \zeroOffset, x={id}, col sep = comma] {\resultsAllTab};
			\addlegendentry{Kernel size (unweighted, \Cref{rule:crown,rule:deg-zero-one-vertices,rule:deg-two-vertices})}
			\addplot[brown,mark=o] table [y expr= 5*(\thisrow{misc_m}-\thisrow{misc_n}-1)+ \zeroOffset, x={id}, col sep = comma] {\resultsAllTab};
			\addlegendentry{Upper-bound of \Cref{thm:fes-lin-kernel}}

			\addplot[orange,mark=square] table [y expr=(\thisrowno{7} + \thisrowno{8}+\zeroOffset), x={id},col sep = comma] {\dataHighWeighted};
			\addlegendentry{Kernel size (weighted, \Cref{rule:deg-zero-weighted,rule:deg-one-weighted,rule:pend-cycle,rule:max-path})}
		\end{axis}
	\end{tikzpicture}
	\vspace{-10pt}
	\caption{Sizes and bounds (in terms of number of vertices plus edges) of various structures and kernels. 
	To all values \zeroOffset{} was added to display 0-values.
	The upper-bound from \cref{thm:crown-exhaustive-application} is not displayed as it only bounds the number of vertices (and has no non-trivial bound on the number of edges).
	The graphs are ordered by relative size of the remaining graph after the data reduction rules (as in \cref{fig:kernel-size2}).
	}
	\label{fig:kernel-vs-input-vs-fes}
\end{figure}
As a matter of fact, we can clearly observe that the upper bound shown in \cref{thm:fes-lin-kernel} 
is not suitable to explain why \cref{rule:deg-zero-one-vertices,rule:deg-two-vertices} 
perform so well in real-world graphs: in 41 of our tested 44 graphs the 
input size itself is already smaller than the guaranteed upper bound of \cref{thm:fes-lin-kernel}.
We exclude \cref{thm:w-kernel} from the discussion here, as the upper bounds given are even weaker than the ones in \cref{thm:fes-lin-kernel}.

\cref{rule:crown} and \cref{thm:crown-exhaustive-application} are better suitable to explain the results:
As can be seen in \Cref{fig:core-vs-size-percentages}, this $2\tau$-bound on the number of vertices is not optimal.
\begin{figure}[t]
	\begin{tikzpicture}[scale=1.0]
		\begin{axis}[
				width=\textwidth,
				height=0.42\textwidth,
				ylabel={size in \%  [$100,\% = n$]},
				legend style = {
						at={(1, 1.02)},
						anchor={south east},
						font = \small
				},
				ymax=160,
				grid,
				legend cell align = left,
				legend columns = 5,
				ytick={0,50,100,150,200},
				xtick=data,
				xticklabels from table={\namesTable}{filename}
				,tick label style={font=\footnotesize}  
				,xmax=2+\TotalRowsResultsAll
				,xmin=0
				,x tick label style={rotate=60,anchor=east}
				]
			\addplot[black,mark=x] table [y expr=100*\thisrow{degeneracy_core_2}/\thisrow{misc_n}, x={id},col sep = comma] {\resultsAllTab};
			\addlegendentry{2-core}
			\addplot[orange!50!yellow,mark=diamond] table [y expr=100*\thisrow{degeneracy_core_3}/\thisrow{misc_n}, x={id},col sep = comma] {\resultsAllTab};
			\addlegendentry{3-core}
			\addplot[blue,mark=triangle] table [y expr=100*\thisrow{k_crown_result_n_part_2}/\thisrow{misc_n}, x={id},col sep = comma] {\resultsAllTab};
			\addlegendentry{Kernel size (\Cref{rule:crown,rule:deg-zero-one-vertices,rule:deg-two-vertices})}
			\addplot[brown,mark=o,discard if={misc_vc_size}{}] table [y expr=200*(\thisrow{misc_vc_size} + 0.0001)/\thisrow{misc_n}, x={id}, col sep = comma] {\resultsAllTab};
			\addlegendentry{$2\cdot$vertex cover number}
		\end{axis}
	\end{tikzpicture}
	\vspace{-10pt}
	\caption{Relative sizes and theoretical upper bounds (in \%; 100\,\% = $n$) of various structures and kernels. The 2-core (3-core) of the graph resulting from iteratively removing all vertices of degree less than 2 (less than 3).
	The vertex cover number was computed with an ILP-solver; on very few graphs we could not compute a minimum vertex cover in a couple of hours and the corresponding graph is skipped in the corresponding plot.
	}
	\label{fig:core-vs-size-percentages}
\end{figure}
However, there is a clear similarity between the lines indicating theoretical upper bound and measured kernel size.
Moreover, on roughly~$2/3$ of the instances the worst-case upper bound~$2\tau$ is smaller than~$n$, that is, the deletion of some vertices is guaranteed.

Overall, \cref{thm:crown-exhaustive-application} seems to deliver the better theoretical explanation.
However, \cref{thm:crown-exhaustive-application} is based on \cref{rule:crown} and, as can be seen in \cref{fig:kernel-size2}, just applying \cref{rule:deg-zero-one-vertices,rule:deg-two-vertices} is almost always better than only applying \cref{rule:crown} (the exceptions are the graphs ``web-NotreDame'', ``web-Stanford'', ``web-Google'', and ``web-BerkStan'').
Thus, while \cref{thm:crown-exhaustive-application} somewhat explains the effects of \cref{rule:crown}, no explanation is provided for the good performance of \cref{rule:deg-zero-one-vertices,rule:deg-two-vertices}.

Summarizing, the current theoretical upper bounds for the kernel size need improvement. 
The most promising route 
seems to be a multivariate analysis in the sense that one should use more than one parameter in the analysis~\cite{Nie10}.
Of course, the challenging part here is finding the ``correct'' parameters.

\subsection{Running times for Maximum-Cardinality Matching}
\label{sec:eval-time-unweighted}
In this section we evaluate the effect on the running time 
of state-of-the-art solvers (see \cref{tab:solvers}) for \CMatch if \cref{rule:crown,rule:deg-zero-one-vertices,rule:deg-two-vertices}
are applied in advance.

Note that all reported running times involving \solverKE are averages over 100 runs where we randomly permute vertex indices in the input.
Although this permutation yields an isomorphic graph, we empirically observed that in the unweighted case the running time of \solverKE heavily depends on the permutation.
For example, choosing a ``good'' or a ``bad'' permutation for the same graph 
may yield speedup of factor 20 or more. 
Precise data on the spectrum of running time variation (for graphs where the time limit was not reached) are shown in \cref{fig:blossom5runningtime}.
\begin{figure}
	\begin{tikzpicture}
		\pgfplotstablegetrowsof{\boxPlotsTab}
		\pgfmathtruncatemacro\TotalRows{\pgfplotsretval-1}
		\pgfmathtruncatemacro\MaxTick{2*\pgfplotsretval + 1.5}
		\pgfplotsset{log base 10 number format code/.code={\pgfmathparse{#1 - 3}$10^{\pgfmathprintnumber{\pgfmathresult}}$}}  %
		\begin{axis}[
				width=\textwidth,
				height=0.42\textwidth,
				grid,
				ymode=log,
				boxplot/draw direction=y,
				ylabel={Time in seconds},
				xticklabels from table={\boxPlotsTab}{filename},
				x tick label as interval=true,
				x tick label style={rotate=40,anchor=north east,font=\footnotesize},
				ytick={1, 10, 100, 1000, 10000, 100000, 1000000},
				xtick={0.5,2.5,4.5,...,\MaxTick},
				xmin=-0.5,
				xmax=1+\MaxTick,
				legend columns=2,
				legend style={
					at={(1,1.02)},
					anchor={south east}
				}
			]
			\pgfplotsinvokeforeach{0,...,\TotalRows}
			{
				\addplot+[
					red,solid,thick,
					boxplot prepared from table={
						table=\boxPlotsTab,
						row=#1,
						lower whisker=blossom_result_time_min,
						upper whisker=blossom_result_time_max,
						lower quartile=blossom_result_time_lower_p,
						upper quartile=blossom_result_time_higher_p,
						median=blossom_result_time_median
					},
					boxplot prepared,
					area legend
				]
				coordinates {};

				\addplot+[
					blue,solid,thick,
					boxplot prepared from table={
						table=\boxPlotsTab,
						row=#1,
						lower whisker=k_blossom_result_time_min,
						upper whisker=k_blossom_result_time_max,
						lower quartile=k_blossom_result_time_lower_p,
						upper quartile=k_blossom_result_time_higher_p,
						median=k_blossom_result_time_median
					},
					boxplot prepared,
					area legend
				]
				coordinates {};

			}
			\addlegendimage{black}
			\addlegendentry{\solverKE (input graph) running time}

			\addlegendimage{blue}
			\addlegendentry{\solverKE (on kernel) running time}
		\end{axis}
	\end{tikzpicture}
	\vspace{-10pt}
	\caption{
		The spectrum of the \solverKE running times
		over 100 runs with random vertex permutations on unweighted graphs.
		The lower (upper) whisker is the minimum (maximum) running time.
		The solid box shows the median and the lower and upper quartile.
		For each graph we have two datasets: left (red) \solverKE on the input graph
		and right (blue) on the kernel.
		We excluded input graphs where we could not perform 100 runs in one hour.	
		To all values \zeroOffset{}\,millisecond was added to display 0-values.
	}
	\label{fig:blossom5runningtime}
\end{figure}
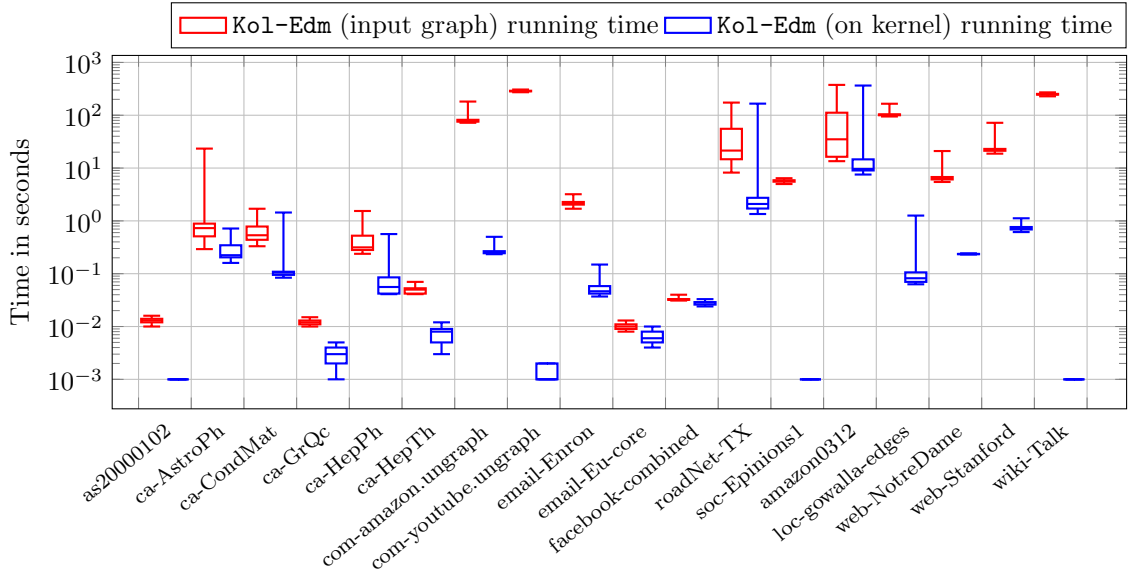%
The running time of all other solvers and our kernelization algorithm were only marginally affected by changing the permutation.

We noticed that the different implementations vary greatly in the time they need for parsing the input graph 
(especially in the smaller graphs \solverHSMV (Python) needs more time to parse the graph 
than \solverKPE (C) needs to find a maximum-cardinality matching).
Moreover, we mainly care about the speedup of the respective algorithms 
(when run on the kernel instead of the original input) and not about the speedup of the graph parsing.
Hence, we neglect the time to parse the input graph in all running-time measures and discussions.
Note that we do \emph{not} neglect the time the implementation needs for parsing the kernel, as this is something that needs only be done with data reduction but not without.
More precisely, for runs without data reduction, 
we report the time the particular implementation needs \emph{after} the graph was loaded; 
for runs with data reduction we report the time of our data reduction rules plus the total time of the implementation 
(including parsing the kernel).

\begin{figure}
	\def\maxValue{10000}
	\def\minValue{0.0005}
	\begin{tikzpicture}[scale=0.95]
		\begin{loglogaxis}[
					width=0.5\textwidth,
					height=0.4\textwidth,
					xlabel={With kernelization [sec]},
					ylabel={Without kernelization [sec]},
					legend style = {
									at={(1, 1.07)},
									anchor={south east},
									font = \small
					},
					legend cell align = left,
					legend columns = 2,
					ytick distance=10^1,
					,xmax=\maxValue
					,xmin=\minValue
					,ymax=\maxValue
					,ymin=\minValue
					]

				\addplot[red,mark=diamond*,only marks,discard if not={blossom_result_n_runs}{100},discard if not={k_blossom_result_n_runs}{100}] table 
				[
				x expr=(\thisrow{k_degree_result_time} + \thisrow{k_crown_result_lp_time} + \thisrow{k_crown_result_time} + \thisrow{k_blossom_result_time_average} + \zeroOffset)/1000,
				y expr=(\thisrow{blossom_running_time_without_k} + \zeroOffset)/1000, col sep = comma] {\resultsAllTab};
			\addlegendentry{\solverKE (no timeouts)}

				\addplot[blue,mark=diamond,only marks,discard if={blossom_result_n_runs}{100},discard if not={k_blossom_result_n_runs}{100}] table 
				[
				x expr=(\thisrow{k_degree_result_time} + \thisrow{k_crown_result_lp_time} + \thisrow{k_crown_result_time} + \thisrow{k_blossom_result_time_average} + \zeroOffset)/1000,
				y expr=(3600), col sep = comma] {\resultsAllTab};
			\addlegendentry{\solverKE (timeouts)}

			\addplot[color=black,domain=\minValue:\maxValue,samples=4] {x};
			\addplot +[color=black,mark=none] coordinates {(\minValue, 3600) (\maxValue, 3600)};
			\addplot +[color=black,mark=none] coordinates {(3600, \minValue) (3600, \maxValue)};

			\addplot[dashed,color=black!75,domain=\minValue:\maxValue,samples=4] {2*x};
			\addplot[dash dot,color=black!75,domain=\minValue:\maxValue,samples=4] {5*x};
			\addplot[dotted,color=black,domain=\minValue:\maxValue,samples=4] {25*x};
			\addplot[dashed,color=black!75,domain=\minValue:\maxValue,samples=4] {0.5*x};
			\addplot[dash dot,color=black!75,domain=\minValue:\maxValue,samples=4] {0.2*x};
			\addplot[dotted,color=black,domain=\minValue:\maxValue,samples=4] {0.04*x};
		\end{loglogaxis}
	\end{tikzpicture}
	\begin{tikzpicture}[scale=0.95]
		\begin{loglogaxis}[
					width=0.5\textwidth,
					height=0.4\textwidth,
					xlabel={With kernelization [sec]},
					ylabel={Without kernelization [sec]},
					legend style = {
									at={(1, 1.07)},
									anchor={south east},
									font = \small
					},
					legend cell align = left,
					legend columns = 2,
					ytick distance=10^1,
					,xmax=\maxValue
					,xmin=\minValue
					,ymax=\maxValue
					,ymin=\minValue
					]
			\addplot[green!50!black,mark=*,only marks,replace if={c_time}{}{8.18868912444}] table 
				[x expr=(\thisrow{k_degree_result_time} + \thisrow{k_crown_result_lp_time} + \thisrow{k_crown_result_time} + \thisrow{k_c_time} + \thisrow{k_c_parsing} + \zeroOffset)/1000,
				y expr=(\thisrow{c_time}+ \zeroOffset)/1000, col sep = comma] {\resultsAllTab};
				\addlegendentry{\solverKPE}

			\addplot[black,mark=x,only marks,replace if={mv_time}{}{8.18868912444}] table 
				[x expr=(\thisrow{k_degree_result_time} + \thisrow{k_crown_result_lp_time} + \thisrow{k_crown_result_time} + \zeroOffset)/1000 + \thisrow{k_mv_time} + \thisrow{k_mv_parsing},
				y expr=\thisrow{mv_time} + (\zeroOffset)/1000, col sep = comma] {\resultsAllTab};
				\addlegendentry{\solverHSMV}

			\addplot[color=black,domain=\minValue:\maxValue,samples=4] {x};
			\addplot +[color=black,mark=none] coordinates {(\minValue, 3600) (\maxValue, 3600)};
			\addplot +[color=black,mark=none] coordinates {(3600, \minValue) (3600, \maxValue)};
			\addplot[dashed,color=black!75,domain=\minValue:\maxValue,samples=4] {2*x};
			\addplot[dash dot,color=black!75,domain=\minValue:\maxValue,samples=4] {5*x};
			\addplot[dotted,color=black,domain=\minValue:\maxValue,samples=4] {25*x};
			\addplot[dashed,color=black!75,domain=\minValue:\maxValue,samples=4] {0.5*x};
			\addplot[dash dot,color=black!75,domain=\minValue:\maxValue,samples=4] {0.2*x};
			\addplot[dotted,color=black,domain=\minValue:\maxValue,samples=4] {0.04*x};
		\end{loglogaxis}
	\end{tikzpicture}
	\caption{Running time of the three state-of-the-art solvers with and without kernelization (each mark indicates one instance). 
	The inclined solid/dashed/dash dotted/dotted lines indicate a factor of 1/2/5/25 difference in the running time.
	To all values \zeroOffset{}\,millisecond was added to display 0-values.
	Timeouts are counted as 1\,h (solid vertical and horizontal line).
	For \solverKE, the timeout behavior is special due to taking the average over 100 runs (see \cref{fig:blossom5runningtime}).
	If at least one of the 100 runs took more than 1\,h, then we aborted the computation and count a timeout (although the average running time might be below 1\,h).
	Hence, we separated the plot for \solverKE (left diagram): 
	The red, filled diamonds indicate instances without any timeout (with or without kernelization).
	The blue, non-filled diamonds indicate instances where without kernelization there was at least one timeout in the 100 runs and with kernelization there was no timeout.
	Thus, the true values for the blue, non-filled diamonds might be below 1\,h (hence the separation).
	(There was no instance where a timeout occurred with kernelization but not without kernelization.)
	}
	\label{fig:new-running-times}
\end{figure}
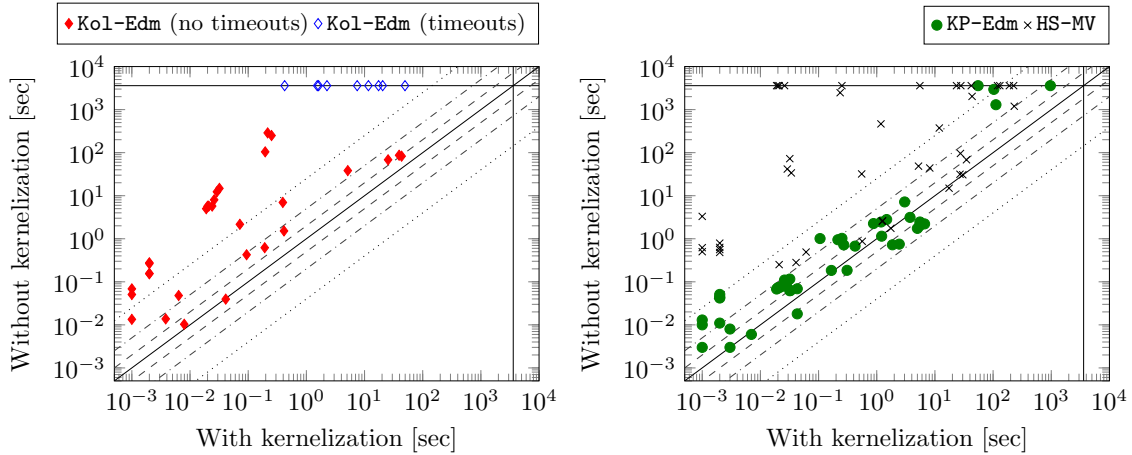

In \cref{fig:new-running-times}, we compare the running times of the solvers when they are directly applied on the input instances
against the running time of the kernelization plus the running time of the solvers on the resulting kernel.
The running time of the \solverKE solver is improved on average by a factor of $157.30$ (median: $29.27$) 
when the kernelization is applied (left-side of \cref{fig:new-running-times}).
The running time of the \solverHSMV solver is improved on average by a factor of $608.79$ (median: $28.87$) when the kernelization is applied (right-side of \cref{fig:new-running-times}).
Hence, it is safe to say that these two algorithms clearly benefit from the kernelization.
However, on the instance \emph{amazon0302} the \solverHSMV solver is 25\% faster without the kernelization and on the instance \emph{facebook-combined} the \solverKE solver is (on average over the hundred runs) 6\% faster without the kernelization.
On all other instances the running times of both solvers where improved by the kernelization.

In contrast, the message drawn by the results for the \solverKPE solver is less clear (right-hand side of \cref{fig:new-running-times}).
The running time of the \solverKPE solver is impaired on 9 of our 44 instances by the kernelization.
On average we \emph{still} get an improvement of the running time by a factor of $4.70$ (median: $2.20$).
The reason for the unclear result for the \solverKPE solver is that most of the instances are too easy for it, 
that is, they are solved very quickly with or without kernelization.
On harder instances (like our largest four graphs) we could observe a more significant speedup gained due to the kernelization.

\subsection{Running times for Maximum-Weight Matching}
\label{sec:eval-time-weighted}
In this section, we evaluate 
\cref{rule:max-path,rule:pend-cycle,rule:deg-one-weighted,rule:deg-zero-weighted}
for \WMatch.
The weighted graphs we used for our tests were generated from the unweighted graphs 
by adding edge-weights between~1 and~1000 chosen independently and uniformly at random.
In the weighted case we tested our kernelization only together with the 
\solverKEW solver since the \solverKPE and \solverHSMV solvers only work in the unweighted setting.
In contrast to the unweighted case (see \cref{fig:blossom5runningtime}),
we could not observe that the running time of \solverKEW is affected when the vertices are permuted.
For consistency, however, we take the average running times also in the weighted case.
Note that for different permutations the data reduction rules were applied in different order resulting in kernels slightly differing in size 
(see \cref{fig:not-confluent} for an example).

\begin{figure}[t]
	\centering
	\begin{tikzpicture}[auto]

		\begin{scope}[scale=0.7]
			\node[knoten] (1) at (0,0) {};
			\node[knoten] (2) at (1,0) {};
			\node[knoten] (3) at (4,0) {};
			\node[knoten] (4) at (2,1) {};
			\node[knoten] (5) at (3,1) {};
			\node[knoten] (6) at (2.5,-1) {};

			\draw[majarr] (1) edge node[above] {$1$} (2);
			\draw[majarr] (2) edge node[above] {$1$} (3);
			\draw[majarr] (2) edge node[left] {$3$} (4);
			\draw[majarr] (4) edge node[above] {$1$} (5);
			\draw[majarr] (5) edge node[right] {$3$} (3);

			\draw[majarr] (2) edge node[left] {$2$} (6);
			\draw[majarr] (6) edge node[right] {$2$} (3);

		\end{scope}

		\draw [->,snake=snake,line after snake=1mm] (3.5,0.0) -- 
			node {\footnotesize $($\Cref{rule:max-path}$,$ \Cref{rule:deg-one-weighted}$)$} (6,0);

		\draw [->,snake=snake,line after snake=1mm] (-0.5,0) -- 
			node[above] {\footnotesize $($\Cref{rule:deg-one-weighted}$,$ \Cref{obs:path-cycle-lin-time}$)$} (-3,0);
			\node at (-4.25,0) {$(\emptyset,\emptyset)$};

		\begin{scope}[yshift=0cm,xshift=6.5cm,scale=0.7]
			\node[knoten] (1) at (0,0) {};
			\node[knoten] (2) at (1,1) {};
			\node[knoten] (3) at (1,-1) {};
			\node[knoten] (4) at (2,0) {};

			\draw[majarr] (1) edge node[above] {$4$} (4);
			\draw[majarr] (1) edge node[left] {$1$} (2);
			\draw[majarr] (1) edge node[left] {$1$} (3);
			\draw[majarr] (2) edge node[right] {$2$} (4);
			\draw[majarr] (3) edge node[right] {$2$} (4);
		\end{scope}
		\begin{scope}[yshift=-2cm]
		
		\begin{scope}[scale=0.7]
			\node[knoten] (1) at (0,0) {};
			\node[knoten] (2) at (1,0) {};
			\node[knoten] (3) at (4,0) {};
			\node[knoten] (4) at (2,1) {};
			\node[knoten] (5) at (3,1) {};
			\node[knoten] (6) at (2.5,-1) {};

			\draw[majarr] (1) edge node[above] {$1$} (2);
			\draw[majarr] (2) edge node[above] {$2$} (3);
			\draw[majarr] (2) edge node[left] {$4$} (4);
			\draw[majarr] (4) edge node[above] {$2$} (5);
			\draw[majarr] (5) edge node[right] {$3$} (3);

			\draw[majarr] (2) edge node[left] {$2$} (6);
			\draw[majarr] (6) edge node[right] {$2$} (3);

		\end{scope}

		\draw [->,snake=snake,line after snake=1mm] (3.5,0.0) -- 
			node {\footnotesize $($\Cref{rule:max-path}$,$ 3x \Cref{rule:deg-one-weighted}$)$} (6,0);

		\draw [->,snake=snake,line after snake=1mm] (-0.5,0) -- 
			node[above] {\footnotesize $($\Cref{rule:deg-one-weighted}$,$ \Cref{rule:max-path}$)$} (-3,0);
			\node at (7.25,0) {$(\emptyset,\emptyset)$};

		\begin{scope}[yshift=0cm,xshift=-5cm,scale=0.7]
			\node[knoten] (1) at (0,0) {};
			\node[knoten] (2) at (1,1) {};
			\node[knoten] (3) at (1,-1) {};
			\node[knoten] (4) at (2,0) {};

			\draw[majarr] (1) edge node[above] {$4$} (4);
			\draw[majarr] (1) edge node[left] {$1$} (2);
			\draw[majarr] (1) edge node[left] {$1$} (3);
			\draw[majarr] (2) edge node[right] {$2$} (4);
			\draw[majarr] (3) edge node[right] {$2$} (4);
		\end{scope}
		\end{scope}
	\end{tikzpicture}
	\caption{In the middle we see two weighted graphs where both \cref{rule:deg-one-weighted} and \cref{rule:max-path} can be applied. 
	On the left-hand side is the resulting graph when \cref{rule:deg-one-weighted} is applied first and
	the right-hand side displays the resulting graph when \cref{rule:max-path} is applied first.}
	\label{fig:not-confluent}
\end{figure}

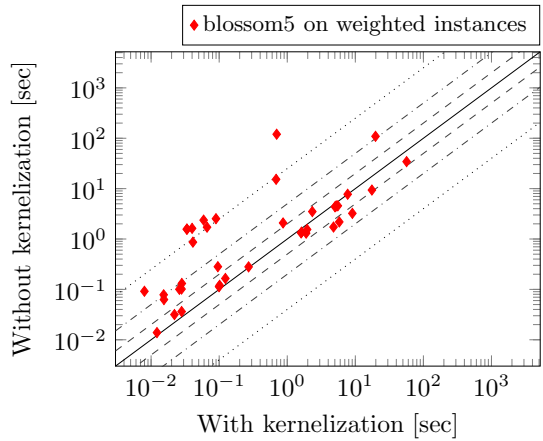
\begin{figure}
		\centering
	\def\maxValue{5200}
	\def\minValue{0.003}
	\begin{tikzpicture}[scale=0.95]
		\begin{loglogaxis}[
					width=0.5\textwidth,
					height=0.4\textwidth,
					xlabel={With kernelization [sec]},
					ylabel={Without kernelization [sec]},
					legend style = {
									at={(1, 1.02)},
									anchor={south east},
									font = \small
					},
					legend cell align = left,
					legend columns = 2,
					ytick distance=10^1,
					,xmax=\maxValue
					,xmin=\minValue
					,ymax=\maxValue
					,ymin=\minValue
					]
 			\addplot[only marks,red,mark=diamond*,discard if={mmtime}{}] table [x expr=(\thisrowno{6} + \thisrowno{10} + \thisrowno{11}+\zeroOffset)*0.001, y expr=(\thisrowno{3}+\thisrowno{4}+\zeroOffset)*0.001,col sep = comma]{\dataHighWeighted};
			\addlegendentry{\solverKEW on weighted instances}
			\addplot[color=black,domain=\minValue:\maxValue,samples=4] {x};
			\addplot[dashed,color=black!75,domain=\minValue:\maxValue,samples=4] {2*x};
			\addplot[dash dot,color=black!75,domain=\minValue:\maxValue,samples=4] {5*x};
			\addplot[dotted,color=black,domain=\minValue:\maxValue,samples=4] {25*x};
			\addplot[dashed,color=black!75,domain=\minValue:\maxValue,samples=4] {0.5*x};
			\addplot[dash dot,color=black!75,domain=\minValue:\maxValue,samples=4] {0.2*x};
			\addplot[dotted,color=black,domain=\minValue:\maxValue,samples=4] {0.04*x};
		\end{loglogaxis}
	\end{tikzpicture}
	\caption{Running time comparison with and without using our 
		kernelization (\Cref{rule:max-path,rule:pend-cycle,rule:deg-one-weighted,rule:deg-zero-weighted}) algorithms before \solverKEW. 
	The solid/dashed/dash dotted/dotted lines indicate a factor of 1/2/5/25 difference in the running time.
	To all values \zeroOffset{}\,millisecond was added to display 0-values.}
	\label{fig:kernel-time}
\end{figure}

In \cref{fig:kernel-time}, we illustrate the running time comparison of our kernelization for the weighted data reduction rules
against various unweighted data reduction rules and the running time of \solverKEW when applied without kernelization 
(the four largest graphs are missing since we could not solve them without kernelization).
Our weighted kernelization algorithm becomes slower than in the unweighted case 
(here, we mean just \cref{rule:deg-zero-one-vertices,rule:deg-two-vertices}).
This is not surprising as our
algorithm for \cref{rule:max-path,rule:pend-cycle,rule:deg-one-weighted,rule:deg-zero-weighted} 
is more involved than the one for the \cref{rule:deg-zero-one-vertices,rule:deg-two-vertices}.
Furthermore, the solver of \citet{Kol09} is significantly faster in the weighted case (\solverKEW) than in the unweighted case (\solverKE).
On three graphs the \solverKEW computes a maximum-weighted matching faster 
than we can produce the kernel.
However, on most graphs, our kernelization algorithm reduces the overall running time of \solverKEW
(on average by a factor of $12.72$; median: $1.40$).
Note that also in the weighted case 
the kernelization is more frequently beneficial than it is not. 

\section{Conclusion}\label{sec:conclusion}
Our work shows that it practically pays off to use (linear-time) data reduction rules for computing maximum (unweighted and weighted) matchings. 
Our current state of the theoretical (kernel size upper bounds)  analysis, however, is insufficient to fully explain this success. 
Here, a multivariate approach in which more than one parameter is taken into consideration seems like the natural next step~\cite{Nie10}.
Finding the \emph{right} parameters is the challenging part here.
In fact, adding to any graph~$G$ to each vertex a new degree-one vertex as neighbor results in a graph~$G'$ where \cref{rule:deg-zero-one-vertices} reduces everything, whereas the original graph~$G$ might not be amenable at all to the data reduction rules. 
Many graph parameters (including feedback edge number) cannot differentiate between~$G$ and~$G'$ and, hence, are not suited to explain the practical effectiveness. 

\paragraph{Future research for unweighted matchings.}
Through the connection between \VC and \CMatch one might be able to transfer further kernelization results from \VC to \CMatch.
However, there are known limitations: obtaining a kernel for \VC with~$O(\tau^{2-\varepsilon})$ edges ($\tau$ is the vertex cover number) for any~$\varepsilon > 0$ is unlikely in the sense that the polynomial hierarchy would collapse~\cite{DM14}.
Thus, obtaining a kernel with~$O(\tau^{2-\varepsilon})$ edges for \CMatch requires a new approach that should \emph{not} work for \VC.
So far, the kernelization algorithms we discuss in this paper \emph{do} work for both \CMatch and \VC.
However, as \CMatch is polynomial-time solvable, an $O(\tau^{2-\varepsilon})$-sized kernel trivially exists for \CMatch. 
The challenge here is to find such a kernel that is (near-)linear-time computable.

In future research, one might also study the combination of data reduction with linear-time approximation algorithms for matching~\cite{DP14}.
Furthermore, it would be interesting to know whether there is an efficient way of applying \cref{rule:relaxed-crown} or a variation of it.

The solver of \citet{Kol09} is significantly faster on weighted graphs (\solverKEW) than on unweighted graphs (\solverKE).
We believe that the reason for this is that in unweighted graphs there are a lot of symmetries, and unlucky tie-breaking seems to have a strong impact on the solver of \citet{Kol09}.
In the weighted case, the performance of the solver of \citet{Kol09} was much more consistent under permuting the vertices in the input graph.
As a consequence, we believe that the following might speedup the algorithm:
given an unweighted graph, introduce edge-weights such that a maximum-weight matching in the then weighted graph is also a maximum-cardinality matching in the unweighted graph.
Using the famous Isolation Lemma \cite{MVV87} one might even enrich and support this with a theoretical analysis.

\paragraph{Future research for weighted matchings.}

While our naive implementation for the unweighted case proved to be quite fast, 
the algorithm for the weighted case could benefit from further tuning.
Note that in the unweighted case \cref{rule:deg-zero-one-vertices,rule:deg-two-vertices} only make changes in the local neighborhood of the affected vertices.
This is not the case in the weighted case, where the application of \cref{rule:deg-one-weighted,rule:max-path,rule:pend-cycle} involve iterations over all edges, see \cref{lem:deg-one-weighted-time,lem:deg-two-lin-time-weighted}.
Hence, applying the data reduction rules exhaustively requires a larger overhead.
Although some improvements in the implementation might be possible, 
an improved algorithmic approach to exhaustively apply the data reduction rules is needed.
Is there a (quasi-)linear-time algorithm to exhaustively apply \cref{rule:deg-zero-weighted,rule:deg-one-weighted,rule:max-path,rule:pend-cycle}?
Furthermore, is there a variant of the crown data reduction (\cref{rule:crown}) for \WMatch?

\medskip

\textbf{Acknowledgement.} 

We are very grateful to anonymous reviewers of ESA~'18 and of ACM JEA for constructive and detailed feedback.

TK was supported by DFG, project FPTinP (NI 369/16).

\bibliographystyle{ACM-Reference-Format}
\bibliography{bib}

\newpage

\appendix %

\section{Appendix - Full Data Set}

\pgfplotstableset{%
	create on use/maxMatch/.style={
		create col/expr={
			\thisrow{k_degree_result_matchings}+\thisrow{k_crown_result_matchings}+\thisrow{k_c_matchings}
		},
	}
}

\begin{table}[h!]\small
	\caption{A full list of graph from the SNAP \cite{snap} data set which we used in our test scenario; here~$|V^{=i}|$ is the number of degree-$i$ vertices, $\Delta$ the maximum degree, and ``MCM'' the cardinality of a maximum matching.}
	\label{tab:unweighted}
	\centering
\pgfplotstabletypeset[columns={filename,misc_n,misc_m,misc_degree_1,misc_degree_2,misc_max_degree,degeneracy_degeneracy,maxMatch},
	columns/filename/.style={string type,column name=Graph,column type = {r}},
	columns/misc_n/.style={column name=$n$,precision=1,column type = {r}},
	columns/misc_m/.style={column name=$m$,precision=1,column type = {r}},
	columns/misc_degree_1/.style={column name=$|V^{=1}|$,precision=1,column type = {r}},
	columns/misc_degree_2/.style={column name=$|V^{=2}|$,precision=1,column type = {r}},
	columns/misc_max_degree/.style={column name=$\Delta$,precision=1,column type = {r}},
	columns/degeneracy_degeneracy/.style={column name=degeneracy,precision=1,column type = {r}},
	columns/maxMatch/.style={column name=MCM,precision=1,column type = {r}},
	every head row/.style ={before row=\toprule, after row=\midrule},
    every last row/.style ={after row=\bottomrule},col sep = comma] {\resultsAllTab}
\end{table}

\end{document}